\documentclass[a4paper,UKenglish]{lipics-v2016}
\usepackage{microtype}
\usepackage[sectionbib]{bibunits}
\newif\ifincludeex
\includeexfalse

\newif\ifnoapp
\noapptrue
\noappfalse

\title{Dynamic Complexity of Parity Games with~Bounded
  Tree-Width\footnote{This 
    work is supported by EU under ERC
    EQualIS (FP7-308087) and STREP Cassting (FP7-601148).}}
\titlerunning{Dynamic Complexity of Parity Games with~Bounded Tree-Width}

\author[1]{Patricia Bouyer}
\author[1]{Vincent Jugé}
\author[1,2]{Nicolas Markey}
\affil[1]{LSV, CNRS \& ENS Cachan, Univ. Paris-Saclay, France}
\affil[2]{IRISA, CNRS \& Inria \& Univ. Rennes 1, France}
\authorrunning{P. Bouyer, V. Jugé and N. Markey} 

\Copyright{Patricia Bouyer, Vincent Jugé and Nicolas Markey}

\usepackage{tikz}
\usetikzlibrary{arrows,backgrounds,calc}
\usepackage{tkz-graph}
\usepackage{enumerate}
\usepackage{fsttcs}
\usepackage{thm-restate}

\EventEditors{}
\EventNoEds{0}
\EventLongTitle{}
\EventShortTitle{}
\EventAcronym{}
\EventYear{}
\EventDate{}
\EventLocation{}
\EventLogo{}
\SeriesVolume{}
\ArticleNo{}

\begin{document}
\begin{bibunit}
\maketitle

\begin{abstract}

Dynamic complexity is concerned with updating the output of a problem
when the input is slightly changed. We~study the dynamic complexity of
two-player parity games over graphs of bounded tree-width, where updates
may add or delete edges, or change the owner or color of states. 
We show that this problem is in \Dyn\FO (with
\LOGSPACE~precomputation); this is achieved by a reduction to a
Dyck-path problem on an acyclic automaton.




\end{abstract}

\section{Introduction}
\label{sec-intro}

\subparagraph{Parity games.}
\looseness=-1 
Games played on graphs are tightly linked with automata and logics. In
those games, two players move a token along the edges of a graph, thus
forming an infinite sequence of states of the graph; each player tries
to make this sequence satisfy her winning condition. Such games
provide a rich and powerful formalism for dealing with numerous
problems in theoretical computer science, especially for verification
and synthesis of reactive systems.

There is a wide range of games on graphs: they can be zero-sum
(involving two players with opposite objectives) or non-zero-sum
(involving several players that may want to collaborate); players may
or may not have full information about the game, and\slash or perfect
observation of the state of the game; strategies may be deterministic
or randomized; and various kinds of winning objectives can be
considered.

\looseness=-1
We focus here on two-player perfect-information turn-based
(zero-sum) parity games. In~such games, the~(finitely many) states of
the graph are labelled with natural numbers; along each infinite play,
the maximal number that appears infinitely many times determines the
winner, one of the player trying to make this number even, and the
other trying to make it odd. Those games have been extensively
studied~\cite{Mos91,EJ91,Zie98,Jur98,FL09a}: deciding the winner in a
parity game admits \NP~and \co\NP~algorithms, but whether this can be
achieved in polynomial time remains one of the foremost open problems
in this area.

\looseness=-1
\subparagraph{Graphs with bounded tree-width.}  The tree-width of
graphs has been defined by Robertson and Seymour~\cite{RS84} as a
measure of the \emph{complexity} of graphs---more precisely, of how
close a graph is to a tree. Many classes of graphs have been shown to
have bounded tree-width~\cite{Bod93}.  Over such graphs, several
(\NP-)hard problems have can be solved efficiently~\cite{Bod93}.
Parity games are such an example where bounded tree-width makes
algorithms more efficient: they have recently been settled
in~\LOGCFL~when the underlying graph has bounded
tree-width~\cite{Gan15}.


\subparagraph{Dynamic problems and dynamic complexity.} 
In this paper, we focus on the \emph{dynamic complexity} of
parity games. Dynamic complexity theory aims at developing algorithms
that are capable of efficiently updating the output of a problem after a
slight change in its input~\cite{PI97,HI02,WS07}. Such algorithms
would keep track of auxiliary information about the current instance,
and update it efficiently when the instance is modified.

\looseness=-1
Consider the problem of reachability in directed graphs,
and equip such graphs with two operations, for respectively inserting
and deleting edges (one~at a~time). It~has recently been proven that this problem
is in the class \Dyn\FO~\cite{DKMSZ15}, which was a
long-standing open problem.  Roughly speaking, an~algorithm is in
\Dyn\FO\ when the solution and the auxiliary information can be
updated by \FO formulas
(or~equivalently, by $\AC^0$ circuits) after a small change in the input.
Moreover, since they maintain some amount of auxiliary information, dynamic
algorithms can be used to compute and maintain witnesses for the problem
they solve---for~reachability, it would maintain a~path between the
source and target nodes, when such a path exists.
%




\looseness=-1
\subparagraph{Our contributions.}
We study parity games from a dynamic perspective, considering the
following basic operations on graphs: insertion and deletion of an
edge, and corruption (change of owner) and change of 
priority of a vertex.
We also assume that we are given a maximal graph, which embeds all
constructed game graphs along the dynamic process: this maximal graph
represents the set of all possible connections in the game.
We~first realize that, by the arguments of~\cite[Corollary 5.7]{PI97},
the parity-game problem over arbitrary
maximal graphs is unlikely  to be solvable in~\Dyn\FO\ (unless $\Dyn\FO = \PTIME$).
We therefore make a standard restriction and assume that the maximal
graph has bounded tree-width. Under this hypothesis, we show that
parity games can be solved in \Dyn\FO\ with \LOGSPACE~precomputation,
that~is, we~can decide whether a distinguished fixed vertex is winning
or not.  To~obtain this result, we~rely on (and~extend) a \Dyn\FO\
algorithm for finding a Dyck path in an acyclic automaton~\cite{WS07},
and build a transformation of our parity-game problem into such a
Dyck-path problem.  The latter transformation is performed by
realizing that local views of possible strategies can be computed
(simple)-step-by-(simple)-step, thanks to the restriction on the
maximal graph. These simple steps can be stored in an
auxiliary graph, in~which only few edges depend on the real edges
that exist in the original game, and the correctness of the
construction goes through the search for paths labelled with Dyck
words.

By lack of space, this extended abstract only focuses on a
simplified framework, where only insertions and deletions of edges are
taken into account.
In the Appendix, we present the proofs for the general
framework, where:
(i)~we~allow
further basic operations like the change of the color or of the owner
of a vertex; (ii)~we~maintain a winning strategy from the
distinguished vertex (if such a strategy exists), and (iii)~we~maintain a uniformly winning
strategy, that~is, a~(memoryless) strategy that is winning from every
winning vertex of the game.  The second point requires maintaining a
witness for the Dyck words problem in the built automaton, while the
third point requires expanding the automaton for the Dyck-word problem
into a non-acyclic automaton.



\subparagraph{Related works.}
\looseness=-1
Given their central role in theoretical computer science, parity games
have been extensively studied in various
contexts~\cite{Zie98,Jur98,DG08a,FL09a}.  This is certainly also due
to the fact that the exact complexity of finding the winner in those
games remains one of the foremost open problems in the area of games on
graphs.
Bounding the tree- or clique-width on the underlying graph
makes the problem easier: a~polynomial-time algorithm was proposed by
Obdr\v z\'alek~\cite{Obd03}, and later improved until the recent
\logcfl~algorithm proposed by Ganardi~\cite{Gan15}.

On the other hand, dynamic complexity is much less developed: while the main 
dynamic complexity classes were defined and studied 20 years
ago~\cite{PI97,HI02,WS07,Zeu15}, only few problems have been
considered from that point of view. As~cited above, directed-graph
reachability has recently been proven in \Dyn\FO, which was an important
open problem in the area. However, the complexity of maintaining a
witness remains an open problem.







\section{Definitions}
\label{sec-defs}

\subsection{Two-player parity games}

A graph is a pair $G = (V,E)$ where $V$ is a finite set of vertices,
and $E \subseteq V \times V$ is a finite set of edges. A
path in the graph~$G$ is a
sequence of vertices $\pi = s_0 \cdot s_1 \cdot s_2 \dots$ such that
for every $i \ge 0$, $(s_i,s_{i+1}) \in E$. If~$\pi$ is a finite path,
we~denote by $\last(\pi)$ the last vertex of~$\pi$.

Let $W \subseteq V$. A~\emph{$W$-selector} in~$G$ is a partial
mapping~$f$ with domain $\mathrm{dom}(f) \subseteq W$
  that associates with every $s \in
\mathrm{dom}(f)$ a vertex $s' \in V$ such that $(s,s') \in E$.
A~path $\pi = s_0 \cdot s_1 \cdot s_2 \dots$ is called compatible with
$f$ whenever $s_i \in \mathrm{dom}(f)$ implies $s_{i+1} =
f(s_i)$. A~maximal such path will be called an \emph{outcome} of~$f$.

A two-player (turn-based) \emph{parity game} is a tuple $G =
(V,E,c,V_0,V_1)$, where $(V,E)$ is a graph, $c\colon V \mapsto \mathbb{N}$
is a coloring function, $V_0 \subseteq V$ is the set of those vertices
that are controlled by $P_0$, and $V_1 = V \setminus V_0$ is the set
of those vertices that are controlled by~$P_1$.

A \emph{strategy} for player $P_p$ ($p\in \{0,1\}$) is then a mapping
$\sigma_p$ that associates, with every finite path~$\pi$ such that
$\last(\pi) \in V_p$, a~vertex~$s'$ such that $(\last(\pi),s') \in E$,
if~any.  The~strategy~$\sigma_p$ is called \emph{memoryless} whenever
there is a $V_p$-selector $f_p$ such that, for every finite
path~$\pi$, $\sigma_p(\pi) = f_p(\last(\pi))$. A~path $\pi = s_0 \cdot
s_1 \cdot s_2 \dots$ is said to be compatible with $\sigma_p$ whenever
$s_0 \dots s_j \in \mathrm{dom}(\sigma_p)$ implies $s_{j+1} = \sigma_p(s_0 \dots s_j)$.
A~maximal compatible path is called an outcome of~$\sigma_p$.

A finite path $\pi = s_0 \cdot s_1 \cdot s_2 \dots s_k$ is winning for
player~$P_p$ if $\last(\pi) \in V_{1-p}$; 
otherwise it is losing for
player~$P_p$. An~infinite path $\pi = s_0 \cdot s_1 \cdot s_2 \dots$
is winning for player~$P_p$ whenever $\limsup_{i \geq 0} c(v_i) \equiv
p \pmod{2}$: $P_0$~wins if the maximal color encountered infinitely
many times is even, and $P_1$~wins otherwise.
Let~$s_0 \in V$ be an initial vertex. A~strategy~$\sigma_p$ for player~$P_p$
is called winning at $s_0$ if every maximal outcome starting at~$s_0$ is winning for
player~$P_p$.

Two-player parity games belong to the Borel hierarchy, hence they are
known for having \emph{uniform memoryless
  strategies}~\cite{EM79}.
%
More precisely, the~set~$V$ of vertices can be split into two subsets $W_0$
and~$W_1$ in such a way that each player~$P_p$ has a memoryless winning
strategy~$f_p$, such that, if~the play starts in some vertex $v \in W_p$, then
player~$P_p$ can ensure winning by following the strategy~$f_p$, no~matter
what her adversary plays: we~say that $\sigma_p$ is an \emph{optimal} strategy
for~$P_p$, and that the corresponding $V_p$-selector~$f_p$ is optimal.

The standard static questions regarding parity games are to decide
whether player $P_0$ has a winning strategy from a distinguished
vertex, and to compute a witnessing winning strategy; or more
generally to compute the set of winning states~$W_0$, and a uniform optimal
winning strategy for player~$P_0$. Characterizing the exact complexity
of these problems is  
a well-known open problem: it~is known that the decision problem is in
$\NP \cap \coNP$~\cite{EM79} and \PTIME-hard, but it is
not known whether it belongs to~\PTIME.







\subsection{Tree decomposition}


The notion of \emph{tree decomposition}~\cite{RS84,RS86} was
introduced by Robertson and Seymour in order to extract classes
of graphs on which problems that are \NP-hard in general might be
tractable.  A~tree decomposition of a graph $G = (V,E)$ 
is
a pair $\mathcal{D} = (\mathcal{T},\mathbf{T})$, where $\mathcal{T} =
(\mathcal{V},\mathcal{E})$ is an undirected tree, and $\mathbf{T}\colon
\mathcal{V} \mapsto 2^V$ is a function such that:
\begin{enumerate}[(i)]
\item for each edge $(s,t) \in E$, there exists a node $v \in
  \mathcal{V}$ such that $\{s,t\} \subseteq \mathbf{T}(v)$;
\item \label{cond2} for each $s \in V$, the set
  $\mathcal{V}_s = \{v \in \mathcal{V} \mid s \in \mathbf{T}(v)\}$ is
  a non-empty connected subset of $\mathcal{T}$.
\end{enumerate}
The \emph{width} of $\mathcal{D}$ is defined as the integer
$\max\{|\mathbf{T}(v)|\mid v \in \mathcal{V}\}-1$, and the
\emph{tree-width} of $G$ is the least width of all tree decompositions of
$G$. 


\subsection{Dynamic complexity classes}

Our aim in this paper is to tackle parity games with a
dynamic-complexity viewpoint. We~briefly introduce the formalisms of
descriptive- and dynamic
complexity here, and refer to~\cite{Imm99,PI97,Hes03} for more details.

\smallskip
Descriptive complexity aims at characterizing positive instances of a problem
using logical formulas: the~input is then described as a logical structure
described by a set of $k$-ary predicates (the~\emph{vocabulary}) over its
universe. For example, a~directed graph can be described by a binary predicate
representing its edges, with the set of states (usually identified to~$[1;n]$
for some~$n$) as the universe. The~problem of
deciding whether each state has at most one outgoing edge can be described by the
first-order formula $\forall x, y, z. (E(x,y) \et E(x,z)) \Rightarrow
(y=z)$. The~class \FO contains all problems that can be characterized by such
first-order formulas. This class corresponds to the circuit-complexity
class~$\AC[0]$ (under adequate reductions)~\cite{BIS90}.

\smallskip
Dynamic complexity aims at developing algorithms that can efficiently update
the output of a problem when the input is slightly changed,
for example reachability of one vertex from another one in a graph.
We~would like our algorithm to take advantage of
previous computations in order to very quickly decide the existence of a path
in the modified graph. 

Formally, a decision problem~$\sfS$ is a subset of the set of
$\tau$-structures $\Struct{\tau}$ built on a vocabulary~$\tau$. In~order to
turn~$\sfS$ into a dynamic problem~$\Dyn\sfS$, we~need to define a
finite set of initial inputs and a finite set of allowed updates.
  For instance, we might use an arbitrary graph as initial input, then
  use a $2$-ary operator $\textsf{ins}(x,y)$ that would insert an edge
  between nodes~$x$ and~$y$. For a universe of size~$n$, the set of
  initial inputs forms a finite alphabet, denoted by $\Xi_n$, and the
  set of update operations forms a finite alphabet, denoted
  by~$\Sigma_n$.
  A~finite, non-empty word in $\Xi_n \cdot \Sigma_n^\ast$ then
  corresponds to a structure obtained by applying a 
  sequence of update operations of~$\Sigma_n$ to an initial structure
  in~$\Xi_n$.  
  The~language~$\Dyn\sfS_n$ is defined as the set of those words in
  $\Xi_n \cdot \Sigma_n^\ast$ 
  that correspond to~structures of~$\sfS$, and
  $\Dyn\sfS$ is the union (over all~$n$) of all such languages.

A~dynamic machine is a uniform family~$(M_n)_{n\in\bbN}$ of deterministic finite automata $M_n=\tuple{Q_n,\Xi_n,\Sigma_n,\delta_n^{\mathrm{init}},\delta_n^{\mathrm{up}},s_n,F_n}$
over initial alphabet~$\Xi_n$ and
update alphabet~$\Sigma_n$, with one initial transition function
$\delta_n^{\mathrm{init}}$ (used when reading the first letter) and
one update transition function $\delta_n^{\mathrm{up}}$ (used when
reading the subsequent letters).  The~set of states can be encoded as a
%
structure over some vocabulary~$\tau^{\mathrm{aux}}$, and corresponds
to a polynomial-size auxiliary data structure. Such a machine solves a
dynamic problem if $\Dyn\sfS_n=\calL(M_n)$ for all~$n$. It~is in the
dynamic complexity class~$\Dyn[\calC']{\calC}$ (or simply~$\Dyn\calC$
if~$\calC=\calC'$) if the update transition function and accepting set
can be computed in~$\calC$, while the initial state and initial
transition function can be computed in~$\calC'$. In other
words, solving any initial instance of the problem (specified as an
element of $\Xi_n$) can be done in $\calC'$, and after any update of
the input (specified by some letter of $\Sigma_n$), further
calculations to solve the problem on that new instance are restricted
to the class $\calC$. Of course, for the a dynamic complexity
class~$\Dyn[\calC']{\calC}$ to have some interest, the class $\calC$
should be easier than the static complexity class of the original
problem.
%

\smallskip In~this paper, we only consider the case where $\calC=\FO$,
meaning that first-order formulas will be used to describe how
predicates are updated along transitions.





\subsection{Main result}

We are now in a position to formally define our problem and state our
main result. We fix a 
positive integer~$\bfC$.  We~follow the approach of~\cite{DG08a} and
assume that parity games are represented as relational
structures $\tuple{V,E,C_1,\ldots,C_{\bfC},V_0,V_1}$, where
$V=V_0\uplus V_1 = \biguplus_{1\leq j\leq\bfC} C_j$, with disjoint
sets~$V_0$ and~$V_1$ and pairwise-disjoint sets~$(C_{j})_{1\leq
  j\leq\bfC}$.  Given a universe $V$, our initial input
  then consists in a tuple $(\sigma,\mathcal{V},\mathcal{C},E_\star)$
  where $\sigma \in V$ is an initial state of the game, $\mathcal{V} =
  (V_0,V_1)$ and $\mathcal{C} = (C_1,\ldots,C_{\bfC})$ are partitions
  of the state set $V$, and $E_\star \subseteq V^2$ is a maximal set
  of edges (they are all given by predicates over universe
    $V$).

  In this extended abstract, we~only focus on the operations of
  insertion and deletion of edges that belong to $E_\star$. More
  precisely, we~let $\Sigma_{E_\star} = \{\textsf{ins}(x,y),
  \textsf{del}(x,y) \mid (x,y) \in E_\star\}$. The effect of a
  sequence of update operations, represented as a
  word~$w\in\Sigma_{E_\star}^\ast$, over a set~$E\subseteq V^2$ of
  edges, is denoted with $w(E)$, and is defined inductively~as:
  \[
  \begin{array}{ll>{\qquad\qquad\qquad}ll}
    E &\quad\text{ if $w=\epsilon$} & 
    E\cup\{(x,y)\} & \quad\text{ if $w=\textsf{ins}(x,y)$} \\
    w'(a(E)) &\quad\text{ if $w=w'\cdot a$} &
    E\setminus\{(x,y)\} &\quad \text{ if $w=\textsf{del}(x,y)$} 
  \end{array}
  \]
  For~$w\in\Sigma_{E_\star}^{\ast}$, we~write $G_w$ for the graph with
  edge set~$w(\emptyset)$ (in~particular, $G_\epsilon$ is the edgeless graph)
  or, by abuse of notation, for the parity game $\tuple{V,w(\emptyset),C_1,\ldots,C_{\bfC},V_0,V_1}$.
  It is to be noted that $G_w$ is a subgraph
  of $(V,E_\star)$. Finally, we~let
  \[
  \Dyn\NonUnifParity_{\mathbf{C}} =
  \{(\sigma,\mathcal{V},\mathcal{C},E_\star) \cdot w \mid w\in\Sigma_{E_\star}^* \text{, Player~$P_0$ wins
    from~$\sigma$ in $G_w$}\}.
  \] 

  As mentioned in the introduction, applying arguments
  of~\cite[Corollary 5.7]{PI97}, it~is unlikely that the above problem
  be solvable in~\PTIME-\Dyn\FO (taking $E_{\star} = V^2$), even with only
  two colors. We therefore
  adopt the idea of bounding the tree-width of graph $(V,E_\star)$.
  In addition to fixing $\bfC$, we also fix a positive integer~$\kappa$
  and we restrict the set of admissible initial inputs:
  the graph $(V,E_\star)$ should be of tree-width at most $\kappa$.
  Hence we restrict the above problem as follows:
  \[
  \Dyn\NonUnifParity_{\mathbf{C},\kappa} =
  \{(\sigma,\mathcal{V},\mathcal{C},E_\star) \cdot w \mid (V,E_\star) \text{ has tree-width at most } \kappa\} \cap
  \Dyn\NonUnifParity_{\mathbf{C}}.
  \]
  Then
  we~design a dynamic algorithm for deciding
  $\Dyn\NonUnifParity_{\mathbf{C},\kappa}$,
  taking as initial graph instance the edgeless graph. Additionally, our
  algorithm maintains information to output a winning strategy for
  player~$P_0$ when such a strategy exists:

  \begin{restatable}{thm}{mainthm}
    \label{mainresult}\label{thm-main}
    Fix positive integers $\bfC$ and $\kappa$.
    The problem $\Dyn\NonUnifParity_{\mathbf{C},\kappa}$
        can be solved in
    $\Dyn[\LOGSPACE]{\FO}$.
\end{restatable}

\looseness=-1
We~give a short overview of the proof here. Our~algorithm consists in
transforming our parity-game problem into an equivalent Dyck-path
problem over a labelled acyclic graph. The~latter problem is known to be
in \Dyn\FO~\cite{WS07} (but we had to adapt the algorithm to our setting).
Our approach for building this acyclic graph 
benefits from
ideas of~\cite{Obd03}:
along some linearization of a tree decomposition of the maximal graph,
we can compute inductively local information about the effects of
strategies.
These computations can be represented as finding a path in an acyclic
graph. However, we~have to resort to \emph{Dyck paths} in order to
make our acyclic graph efficiently updatable when the parity game is
modified.

\section{Precomputation: nice tree decompositions and depth-first traversals}
\label{sec-treedec}

Before looking into parity games themselves, we introduce a special kind
of rooted tree decomposition, called \emph{nice
  decomposition}~\cite{Klo94}.

\begin{dfn}
  Let $\mathcal{D} = (\mathcal{T},\mathbf{T})$ be a tree decomposition
  of a graph $G = (V,E)$, with $\mathcal{T} = (\mathcal{V},\mathcal{E})$.
  We say that the tree decomposition $\mathcal{D}$ is \emph{nice} if:
  \begin{itemize}
   \item for all $s \in V$, there exists a node $v_s$ of $\mathcal{T}$
   such that $\mathbf{T}(v_s) = \{s\}$;
   \item for all leaves $v$ of $\mathcal{T}$, there exists a vertex $s \in V$ such that $v = v_s$;
   \item for all edges $(v,w)$ of $\mathcal{T}$, either
   $|\mathbf{T}(v)| = |\mathbf{T}(w)|+1$ and $\mathbf{T}(v) \supseteq \mathbf{T}(w)$
   or $|\mathbf{T}(v)|+1 = |\mathbf{T}(w)|$ and $\mathbf{T}(w) \subseteq \mathbf{T}(w)$.
  \end{itemize}
  
  In addition, for all vertices $s \in V$, we abusively say that $\mathcal{D}$ is
  \emph{rooted at $s$} if we root the tree $\mathcal{T}$ at the node $v_s$.
\end{dfn}
Nice decompositions ensure very small changes between adjacent nodes
of the tree, which is somehow required in a dynamic-complexity
perspective.

Extending~\cite{EJT10}, one can show the following
(see~Appendix~\ref{app-pureTD}):
\begin{restatable}{lem}{puredecomposition}
  \label{lemma:puredecomposition} Let $G$ be a graph with $n$ vertices
  and with tree-width~$\kappa$.  We~can construct in \LOGSPACE\ a nice
  tree decomposition $(\mathcal{T},\mathbf{T})$ of~$G$, with width at
  most $4 \kappa+3$ and diameter at most $c(\kappa) \cdot
  (\log_2(n)+1)$, where $c(\kappa)$
  only depends on~$\kappa$.
\end{restatable}

However, we will not use nice tree decompositions directly.
Instead, we first need to perform \emph{depth-first traversals} of such decompositions.

\begin{dfn}
  Let $\mathcal{T}$ be a rooted tree, with root $\rho$.  For each
  vertex $v$ of $\mathcal{T}$, we denote by $\mathcal{A}(v)$ the set
  of ancestors of $v$ in $\mathcal{T}$, including $v$ itself.  A
  \emph{depth-first traversal} of $\mathcal{T}$ is a sequence
  $v_1,\ldots,v_\ell$ of vertices of $\mathcal{T}$ such that
  \begin{itemize}
  \item $v_1 = v_\ell = \rho$;
  \item all vertices of $\mathcal{T}$ belong to the set $\{v_1,\ldots,v_\ell\}$;
  \item for all $i \leq \ell-1$, $v_{i+1}$ is either a child or a parent of $v_i$;
  \item for all $x \leq y \leq z \leq \ell$, we have $\mathcal{A}(x) \cap \mathcal{A}(z) \subseteq \mathcal{A}(y)$.
  \end{itemize}
\end{dfn}

Performing a depth-first traversal on a rooted tree-decomposition of a
graph $G$ is a standard construction of a \emph{linear decomposition}
of~$G$, i.e., a~degenerate rooted tree decomposition of~$G$ in which
every internal node has one unique child~\cite{RS86}.  More precisely,
the sequence $(\mathcal{A}(v_i))_{1 \leq i \leq \ell}$ is a linear
decomposition of $\mathcal{T}$, from which a linear decomposition
of~$G$ is easily derived.  However, in the sequel, we do \emph{not}
focus on the latter linear decomposition, but rather on a nice
tree-decomposition $\mathcal{T}$ of $G$ and on a depth-first traversal
of~$\mathcal{T}$ itself.

These constructions are illustrated in
Fig.~\ref{fig:etendue-0}, which displays 
a graph $G$, a nice tree decomposition $(\mathcal{T},\mathbf{T})$ of~$G$,
and depth-first traversal of~$\mathcal{T}$.

\begin{figure}[ht]
\begin{center}
\begin{tikzpicture}[scale=0.4]
\SetGraphUnit{5}
\Vertex[x=-9,y=5.5,L={$s_1$}]{1}
\SO[L={$s_2$}](1){2}
\EA[L={$s_3$}](1){3}
\SO[L={$s_4$}](3){4}
\node[anchor=south] at (-6.5,6.5) {Graph $G$};

\Edge[style={<->}](1)(2)
\Edge[style={<-}](1)(3)
\Edge[style={->}](1)(4)
\Edge[style={<->}](2)(4)
\Edge[style={<-}](3)(4)
\Loop[dist=2cm,dir=EA,style={->}](3)
 
\SetGraphUnit{2}

\placenode{0}{0}{g}{7}
\Vertex[Node,L={$s_1$}]{1}
\Vertex[Node,L={$s_2$}]{2}
\Vertex[Node,L={$s_4$}]{4}
\Edge(1)(2)
\Edge(1)(4)
\Edge(2)(4)

\placenode{0}{4.4}{h}{8}
\Vertex[Node,L={$s_1$}]{1}
\Vertex[Node,L={$s_2$}]{2}
\Edge(1)(2)

\placenode{0}{8.8}{i}{9}
\Vertex[Node,L={$s_2$}]{2}

\placenode{6}{0}{f}{6}
\Vertex[Node,L={$s_4$}]{4}

\placenode{6}{4.4}{b}{2}
\Vertex[Node,L={$s_1$}]{1}
\Vertex[Node,L={$s_4$}]{4}
\Edge(1)(4)

\placenode{6}{8.8}{a}{1}
\Vertex[Node,L={$s_1$}]{1}

\placenode{12}{0}{c}{3}
\Vertex[Node,L={$s_1$}]{1}
\Vertex[Node,L={$s_3$}]{3}
\Vertex[Node,L={$s_4$}]{4}
\Edge(1)(3)
\Edge(1)(4)
\Edge(3)(4)

\placenode{12}{4.4}{d}{4}
\Vertex[Node,L={$s_1$}]{1}
\Vertex[Node,L={$s_3$}]{3}
\Edge(1)(3)

\placenode{12}{8.8}{e}{5}
\Vertex[Node,L={$s_3$}]{3}

\Edge(N7)(S8)
\Edge(N8)(S9)
\Edge(NE7)(SW2)
\Edge(N6)(S2)
\Edge(N2)(S1)
\Edge(NW3)(SE2)
\Edge(N3)(S4)
\Edge(N4)(S5)

\node[anchor=south] at (N1) {\textbf{ROOT}};

\node[anchor=west] at (16.2,3.4){
\addtolength{\tabcolsep}{1pt}   
\begin{tabular}{ll}
\multicolumn{2}{c}{Depth-first} \\
\multicolumn{2}{c}{traversal:} \\
$v_1 = a$ & $v_2 = b$ \\ $v_3 = c$ &
$v_4 = d$ \\ $v_5 = e$ & $v_6 = d$ \\
$v_7 = c$ & $v_8 = b$ \\ $v_9 = f$ &
$v_{10} = b$ \\ $v_{11} = g$ & $v_{12} = h$ \\
$v_{13} = i$ & $v_{14} = h$ \\ $v_{15} = g$ &
$v_{16} = b$ \\ $v_{17} = a$ &
\end{tabular}
\addtolength{\tabcolsep}{-1pt}   
};

\end{tikzpicture}
\end{center}
\caption{Graph, nice tree decomposition and depth-first traversal}
\label{fig:etendue-0}\label{ex-xTD-0}
\end{figure}
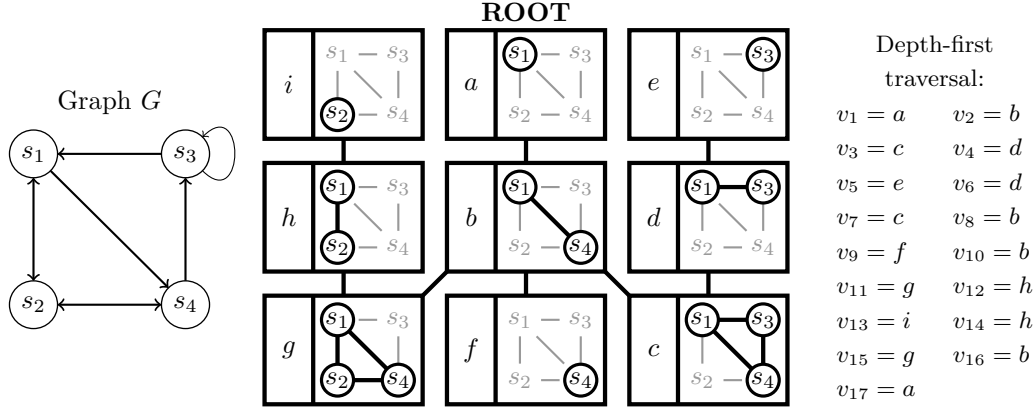

%
%
%


\section{Towards a dynamic algorithm: selector views}
\label{sec-stratview}


In this section, we fix an input game $G$ and we assume that we have a
nice tree decomposition $\mathcal{D} = (\mathcal{T},\mathbf{T})$ of $G$
such as constructed in Lemma~\ref{lemma:puredecomposition},
rooted at a vertex $\sigma$ of $G$, and that we have a
depth-first traversal $v_1,\ldots,v_\ell$ of $\mathcal{T}$. We
transform our parity-game problem into a Dyck-path problem.
More precisely, we build a
graph~$\Gamma_G$, and establish a correspondence between some Dyck
paths in this graph and winning strategies in the original game.

Moreover, in order to simplify subsequent proofs, and without loss of
generality, we assume from this point on that $G$ contains at least
two vertices.

\subsection{Selector views}\label{sec-selview}

For all nodes $v$ of $\mathcal{T}$,
recall that $\mathcal{A}(v)$ is the set of ancestors of $v$.
For all integers $i \in [1,\ell]$, we define the sets $\Theta_i = \mathcal{A}(v_i)$
and $\Psi_i = \bigcup_{v\in\Theta_i} \bfT(v)$. Then, for all $i \in [0,\ell]$, we also set
\begin{xalignat*}4
\Theta_{\leq i} &= \textstyle\bigcup_{j=1}^i \Theta_j &
   \Theta_{>i} &= \calV \setminus \Theta_{\leq i} &
  \Psi_{\leq i} &= \textstyle\bigcup_{j=1}^i \Psi_j &
   \Psi_{>i} &= V \setminus \Psi_{\leq i} 
\end{xalignat*}
%
Observe that
$\Psi_{\leq 0} = \Psi_{>\ell} = \emptyset$ and that $\Psi_{\leq 1} = \bfT(v_1) = \{\sigma\}$.

By definition, for all $1 \leq i \leq \ell$, the sets $\Theta_i$ and $\Theta_{\leq i}$
are non-empty and ancestor-closed.
Therefore, each node $w \in \Theta_{>i-1}$ has a unique lowest ancestor in~$\Theta_{\leq i}$,
which we denote by~$n_i(w)$. Let $a \leq i$ and $b \geq i$ be integers such that
$n_i(w) \in \Theta_a = \calA(v_a)$ and that $w = v_b$.
Thus $n_i(w)$ appears in both $\mathcal{A}(v_a)$ and $\mathcal{A}(v_b)$,
which entails that $n_i(w) \in \mathcal{A}(v_i) = \Theta_i$.

Furthermore, for each vertex $s \in V$, we let
$\calT^s=\{v\in\calV\mid s\in \bfT(v)\}$.  Denoting by~$r(s)$ the root
of the sub-tree $\mathcal{T}^s$ of $\mathcal{T}$, we know that $s \in
\Psi_{>i-1}$ if, and only~if, $r(s) \in \Theta_{>i-1}$; in~that case,
we~also denote by~$n_i(s)$ the node $n_i(r(s))$, and it is easy to see
that $n_i(v) = n_i(s)$ for all nodes $v \in \mathcal{T}^s$.

In what follows, for all nodes $v \in \Theta_i$, we denote by
$\Theta_{>i-1}(v)$, $\Psi_{>i-1}(v)$ the sets $\{w \in \Theta_{>i-1}
\mid n_i(w) = v\}$ and $\{s \in \Psi_{>i-1} \mid n_i(s) = v\}$, which
respectively form partitions (possibly containing empty sets) of the
sets $\Theta_{>i-1}$ and $\Psi_{>i-1}$.  Then every edge of~$G$ with
one end~$s$ in~$\Psi_{>i-1}(v)$ has its other end in $\Psi_{>i-1}(v)
\cup \bfT(v)$, since $\Theta_{>i-1}(v)\cup \{v\}$ is a subtree of
$\mathcal{T}$ that contains $\calT^s$.
%
%
%
%
Figure~\ref{fig:partition-theta-0} shows some sets $\Theta_{>i}(v)$ and $\Psi_{>i}(v)$
obtained from the depth-first traversal of Figure~\ref{fig:etendue-0}.

\begin{figure}[tb]
\begin{minipage}{.4\linewidth}
\begin{center}
\begin{tikzpicture}[scale=0.5]
\draw[draw=black!60,fill=black!5] (-0.3,0.7) -- (2.3,0.7) -- (2.7,0.3) --
(2.7,-4.3) -- (2.3,-4.7) -- (1.7,-4.7) -- (1.3,-4.3) -- (1.3,-3.7) -- (0.3,-2.7) --
(-0.3,-2.7) -- (-0.7,-2.3) -- (-0.7,0.3) -- cycle;
\draw[draw=black!60,fill=black!5] (-2.3,0.7) -- (-1.7,0.7) -- (-1.3,0.3) --
(-1.3,-2.9) -- (-0.9,-3.3) -- (0.3,-3.3) -- (0.7,-3.7) -- (0.7,-4.3) -- (0.3,-4.7) --
(-2.3,-4.7) -- (-2.7,-4.3) -- (-2.7,0.3) -- cycle;

\SetGraphUnit{2}

\tikzset{VertexStyle/.style = {color=black}}
\vs

\Vertex[L={$a$}]{1}
\SO[L={$b$}](1){2}

\tikzset{VertexStyle/.style = {color=black}}

\WE[L={$i$}](1){9}
\SO[L={$h$}](9){8}
\SO[L={$g$}](8){7}
\SO[L={$f$}](2){6}
\EA[L={$e$}](1){5}
\SO[L={$d$}](5){4}
\SO[L={$c$}](4){3}

\Edge(1)(2)
\Edge(2)(3)
\Edge(3)(4)
\Edge(4)(5)
\Edge(2)(6)
\Edge(2)(7)
\Edge(7)(8)
\Edge(8)(9)

\node[anchor=north] at (2,-4.7) {$\Theta_{\leq 7}$};
\node[anchor=north] at (-1,-4.7) {$\Theta_{>7}(b)$};

\node[anchor=west] at (3.1,-0.5) {$\Theta_{8} = \{a,b\}$};
\node[anchor=west] at (3.1,-1.5) {$\Theta_{>7}(a) = \emptyset$};
\node[anchor=west] at (3.1,-2.5) {$\Psi_{>7}(a) = \emptyset$};
\node[anchor=west] at (3.1,-3.5) {$\Psi_{>7}(b) = \{s_2\}$};

\node at (0,1.1) {};
\end{tikzpicture}
\end{center}
\caption{Sets $\Theta_{>i}(v)$ and $\Psi_{>i}(v)$}
\label{fig:partition-theta-0}\label{ex-psi-0}
\end{minipage}\hfill
\begin{minipage}{.58\linewidth}
\begin{center}
\begin{tikzpicture}[scale=1]
\draw[draw=black!60,thick,fill=black!5] (-0.6,0.4) -- (-0.4,0.6) -- (0.4,0.6) -- (0.6,0.4) -- (0.6,-0.4) -- (1.1,-0.9) --
(1.9,-0.9) -- (2.1,-1.1) -- (2.1,-1.9) -- (1.9,-2.1) -- (1.1,-2.1) -- (0.9,-1.9) -- (0.9,-1.1) -- (0.4,-0.6) --
(-0.4,-0.6) -- (-0.6,-0.4) -- cycle;
\draw[draw=black!75,thick,fill=black!15] (-0.5,0.36) -- (-0.36,0.5) -- (0.36,0.5) -- (0.5,0.36) -- (0.5,-0.36) -- (0.36,-0.5) -- (-0.36,-0.5) -- (-0.5,-0.36) -- cycle;
\draw[draw=black!60,thick,fill=black!5] (-0.6,-1.1) -- (-0.4,-0.9) -- (0.4,-0.9) -- (0.6,-1.1) -- (0.6,-1.9) -- (0.4,-2.1) -- (-0.4,-2.1) -- (-0.6,-1.9) -- cycle;
\draw[draw=black!75,thick] (-0.6,0) -- (-1.2,0) -- (-1.2,-0.2);
\node[anchor=north] at (-1.2,-0.1) {$\mathbf{T}(a)$};
\node[anchor=north] at (0,-2.1) {$\Psi_{>7}(b)$};
\node[anchor=north] at (1.5,-2.1) {$\mathbf{T}(b)$};

\SetGraphUnit{1.5}
\vs
\Vertex[L={$s_1$}]{1}
\SO[L={$s_2$}](1){2}
\EA[L={$s_3$}](1){3}
\SO[L={$s_4$}](3){4}

\tikzset{EdgeStyle/.append style = {>=stealth}}

\Edge[style={<->,ultra thick}](1)(2)
\Edge[style={<-,densely dotted}](1)(3)
\Edge[style={->,densely dotted}](1)(4)
\Edge[style={<-,ultra thick}](2)(4)
\Edge[style={->,bend right=40,densely dotted}](2)(4)
\Edge[style={<-,densely dotted}](3)(4)
\Loop[dist=1cm,dir=EA,style={ultra thick,->}](3)

\node[anchor=west] at (2.7,0.5) {$c(s_1) = 1$};
\node[anchor=west] at (4.5,0.5) {$c(s_2) = 2$};
\node[anchor=west] at (2.7,0) {$c(s_3) = 1$};
\node[anchor=west] at (4.5,0) {$c(s_4) = 3$};

\node[anchor=west] at (2.7,-1) {$f_8(a,s_1) = \mathord?$};
\node[anchor=west] at (2.7,-1.5) {$f_8(b,s_1) = (s_1,2)$};
\node[anchor=west] at (2.7,-2){$f_8(b,s_4) = (s_1,3)$};
\end{tikzpicture}
\end{center}
\caption{Selector views~$f_i(v)$ of a $V$-selector~$f$}
\label{fig:strategie-projection-0}\label{ex-SV-0}
\end{minipage}
\end{figure}

\medskip
It is worth noting that $\Psi_{>i}(\cdot)$ only slightly differs
  from $\Psi_{>i+1}(\cdot)$. 
  This will allow us to propagate useful information on $V$-selectors in a backward
  manner, from $\Psi_{>\ell}$
  to $\Psi_{>0}$.
In order to define the information we need about $V$-selectors, 
we~introduce three new symbols~$?$, $\top$ and~$\bot$,
and for every node $v \in \mathcal{V}$, we~denote by~$\Lambda(v)$ the
set of functions $\mathbf{T}(v) \to \{?,\top,\bot\} \cup (\bfT(v)
\times \{1,\ldots,\mathbf{C}\})$.
With every $V$-selector~$f$, every index $1 \le i \le \ell$, and every node $v
  \in \Theta_i$, we associate such a function, which encodes
  local
  information about the behavior of~$f$ on~$\Psi_{>i-1}(v)$.
We make it explicit now.

\begin{dfn}\label{dfn:selector-view-f}
  Consider an integer $1\leq i \leq \ell$, a node $v \in \Theta_i$,
  a vertex $s \in \mathbf{T}(v)$ and a $V$-selector~$f$. Let $\pi = s \cdot
  s_1 \cdot s_2 \dots$ be the unique (maximal) outcome of~$f$ from~$s$.
  
  We call \emph{selector view}, and denote by $f_i(v, s)$, the element
  of $\{?,\top,\bot\} \cup (\bfT(v) \times \{1,\ldots,\mathbf{C}\})$
  defined as follows:
  \begin{itemize}
  \item if $\pi$ stays in~$\Psi_{>i-1}(v)$ after the first step (or if
    $s \in \Psi_{>i-1}(v) \setminus \mathrm{dom}(f)$) and is winning
    for~$P_0$, then $f_i(v,s) = \top$;
  \item if $\pi$ stays in~$\Psi_{>i-1}(v)$ after the first step (or if
    $s \in \Psi_{>i-1}(v) \setminus \mathrm{dom}(f)$) and is losing
    for~$P_0$, then $f_i(v,s) = \bot$;
  \item if the first step of $\pi$ does not enter $\Psi_{>i-1}(v)$ (or
    if $s \notin \Psi_{>i-1}(v)$ and $s \notin \mathrm{dom}(f)$), then
    $f_i(v,s) = \mathord?$;
  \item if $\pi$ stays in~$\Psi_{>i-1}(v)$ from the first step until
    some step~$j\geq 2$, and $c_m$~is the maximal color
    seen along~$\pi$ between~$s$ and~$s_{j-1}$, then $f_i(v,s) =
    (s_j,c_m)$.
  \end{itemize}
  By extension, we also call selector views the functions $f_i(v) : s
  \mapsto f_i(v,s)$.
\end{dfn}

%
%
%

Figure~\ref{fig:strategie-projection-0} displays examples of selector views.
The~$V$-selector~$f$ is represented by solid edges, and dotted edges are
those graph edges that are not used by~$f$. The~color~$c(s_i)$ of
each vertex $s_i \in V$ is indicated on the right.

%
%
%
%
%

We~explain in the next section how those views can be computed
iteratively by efficient local 
operations.  We~then proceed in the same way with $V_0$-strategies
(which we target); this will provide us with a graph in which a simple
Dyck-path query will answer our problem.

\subsection{Computing selector views}\label{sec-refselview}


We aim at computing the sequence of selector views $(f_i)_{1\leq i\leq \ell}$ in a
backward manner, computing $(f_i(v))_{v\in\Theta_i}$
from~$(f_{i+1}(v))_{v\in\Theta_{i+1}}$. However, knowing the latter views is not
sufficient for reconstructing the former ones: intuitively speaking, the
functions $f_{i+1}(v)$ provide us with knowledge about the behavior
of~$f$ in each domain~$\Psi_{>i}(v)$, i.e.~in~$\Psi_{>i}$, and we need
to get information on how $f$ behaves on~$\Psi_{>i-1}$. When $\Psi_{>i}$
is a strict subset of~$\Psi_{>i-1}$, we~need some additional
information, as~we explain~now.

Fix an index~$i$ ($1 \le i \le \ell$) such that $\Psi_{>i} \subsetneq
\Psi_{>i-1}$. Such an index~$i$ is said to be
\emph{critical}. It~corresponds to steps of the depth-first traversal
where a new node is being explored.  More precisely, this means that
the node $v_i$ is a child of $v_{i-1}$ (or that $i = 1$) and that
there exists a (necessarily unique) vertex $\theta_i$ of $G$ such that
$v_i = r(\theta_i)$.  Then, we have $\Psi_{>i-1} = \Psi_{>i} \uplus \{\theta_i\}$.

Consequently, in order to know ``how~$f$ behaves on~$\Psi_{>i-1}$'',
we~need to identify which vertex $s \in \mathbf{T}(v_i)$, if~any, is such that
$f(\theta_i) = s$, and which vertices $t \in \mathbf{T}(v_i)$
are such that $f(t) = \theta_i$.


One can show the following result, whose proof is postponed to
  Appendix~\ref{section:decremental-I}:
\begin{restatable}{proposition}{propdecrI}
    \label{prop:decremental-I}
    Fix a $V$-selector~$f$. For all integers $i \leq \ell$, it holds:
    \begin{itemize}
    \item if $i = \ell$, then $f_i$ is \FO-definable;
    \item if $i \leq \ell-1$ is non-critical, then $f_i$ is \FO-definable using $f_{i+1}$;
    \item if $i \leq \ell-1$ is critical, then $f_i$ is \FO-definable using
      $f_{i+1}$ and a predicate that indicates which
      vertices $s \in \mathbf{T}(v_i)$ are such that $f(\theta_i) =
      s$, and which vertices $t \in \mathbf{T}(v_i)$ are such that
      $f(t) = \theta_i$.
    \end{itemize}
  %
\end{restatable}

The correctness of the approach is then straightforward:
\begin{lem}\label{lem-wintop}
  Let $\sigma$ be the vertex of $G$ at which the nice tree decomposition $\mathcal{T}$ is rooted:
  we have $\Theta_1 = \{v_1\}$ and $\mathbf{T}(v_1)=\{\sigma\}$.
  Furthermore, let~$\pi = \sigma \cdot s_1 \cdot s_2 \dots$ be the unique
  (maximal) outcome of~$f$ from~$\sigma$. Then
  (a)~$f_1(v_1,\sigma) \in \{\bot,\top\}$; and (b)~$\pi$~is winning for
  $P_0$ if, and only~if, $f_1(v_1,\sigma) = \top$.
\end{lem}

\subsection{Views of \texorpdfstring{\(V_0\)}{V0}-selectors}
\label{section:decremental-II}

Let $g$ be a $V_0$-selector and let~$f$ be a $V$-selector. We~say that
$f$ is an \emph{extension} of~$g$ if $g$ is the restriction of~$f$ to
the domain~$V_0$. For all $i$, we would like to compute set
  $\{f_i \mid f \text{ is an extension of } g \}$. 
  One may wish to view each set as a subset of $\prod_{v \in \Theta_i}
  \Lambda(v)$.  Unfortunately, this latter set may be polynomially
  large, since $\Theta_i = \calA(v_i)$ may have size up to $c(\kappa)
  (\log_2(n)+1)$.  Hence, no \FO formula with domain $V$ and with a
  given arity can distinguish all the subsets of $\prod_{v \in
    \Theta_i} \Lambda(v)$ when $n$ grows, and we shall use a less
  direct approach.
%
%
%

We over-approximate those sets by projecting over every $v
\in \Theta_i$ as follows.
Let $\Lambda_i = \prod_{v \in \Theta_i} 2^{\Lambda(v)}$.
We call \emph{selector view} of the $V_0$-selector~$g$, and denote by~$\omega_i^g$,
the~element of $\Lambda_i$ defined~by: 
$\omega_i^g : v \mapsto = \{f_i(v) \mid f \text{ is an
  extension of } g\}$.
Since every set $\Lambda(v)$ has a size at most
$(3 + 4(\kappa+1)\mathbf{C})^{4(\kappa+1)}$ and
since $|\Theta_i|\leq c(\kappa)\cdot (\log_2(n)+1)$,
the cardinality of $\Lambda_i$ is polynomial in~$n$.

  %

We~then argue that selector views of $V_0$-selectors are
sufficient for deciding whether the vertex~$\sigma$ is winning
for~$P_0$, and that they can be computed in a backward
manner by \FO formulas.

\begin{restatable}{lem}{lemvictory}
  \label{lem:victory}
  The vertex $\sigma$ is winning for $P_0$ if, and only~if, there exists
  a $V_0$-selector~$g$ such that $\omega_1^g(v_1)$ is the
  singleton set $\{\sigma \mapsto \top\}$.
\end{restatable}


Moreover, for all $V$-selectors~$f$, the collection
$(f_i(v))_{v \in \Theta_i}$ 
satisfies the 
following \emph{compatibility} requirement: 
\[
\forall v,w\in\Theta_i.\ \bigl[
(v\not=w)  \Rightarrow
\forall s\in \bfT(v)\cap \bfT(w).\
\mathord?\in\{f_i(v,s), f_i(w,s)\}\bigr].
\]
  This is because $\Psi_{>i}(v)
  \cap \Psi_{>i}(w) = \emptyset$. We therefore denote by
  $\mathcal{C}^0$ the set of those tuples $(\varphi_v)_{v \in
  X}$ with $X \subseteq \mathcal{T}$ and $\varphi_v \in \Lambda(v)$ for all $v \in X$,
  such that for all nodes $v, w \in X$, if $v \neq w$, then for every $s \in
  \mathbf{T}(v) \cap \mathbf{T}(w)$, $\mathord?\in
  \{\varphi_v(s),\varphi_w(s)\}$.
%
%
For all $V_0$-selectors~$g$, the selector views
$\omega_{i}^g$ enjoy a similar compatibility requirement.

\begin{restatable}{lem}{lemsoundness}
  \label{lem:soundness}
  Let $g$ be a $V_0$-selector, and let $X$ be a subset of $\Theta_i$.
  For every tuple $(\varphi_v)_{v \in X}$ in $\mathcal{C}^0 \cap \prod_{v \in X}
  \omega_{i}^g(v)$, there exists an extension $f$ of $g$ such that
  $\varphi_v = f_{i}(v)$ for all nodes $v \in X$.
\end{restatable}
Thanks to Lemma~\ref{lem:soundness}, Proposition~\ref{prop:decremental-I}
leads to a similar result for $V_0$-selectors, whose proof is given
in Appendix~\ref{app:decremental-II}.
\begin{restatable}{proposition}{propdecrII}
    \label{prop:decremental-II}
    Fix a $V_0$-selector~$g$.  For all integers $i \leq \ell$,
    there exists a \FO-definable function $\Omega_i$ such that:
  \begin{itemize}
  \item if $i = \ell$, then $\Omega_i$ is a constant, and $\omega_i^g = \Omega_i$;
  \item if $i \leq \ell-1$ is non-critical, then $\omega_i^g = \Omega_i(\omega_{i+1}^g)$;
  \item if $i \leq \ell-1$ is critical, then $\omega_i^g = \Omega_i(\omega_{i+1}^g,(\mathbf{B}_{\theta_i,s})_{s \in \mathbf{T}(v_i)},(\mathbf{B}'_{s,\theta_i})_{s \in \mathbf{T}(v_i)})$,
  where $(\mathbf{B}_{\theta_i,s})_{s \in \mathbf{T}(v_i)}$ and $(\mathbf{B}'_{s,\theta_i})_{s \in \mathbf{T}(v_i)}$ are vectors whose entries are defined by:
      \[
      \mathbf{B}'_{x,y} = \mathbf{B}_{x,y} = \left\{\begin{array}{ll} \{\bot\} & \text{if}\
          (x,y) \notin E,\ \text{or}\ x \in V_0 \setminus \mathrm{dom}(g), \\
          & \phantom{\text{if}\ (x,y) \notin E,\ {}}\text{or}\ x \in
          \mathrm{dom}(g)\ \text{and}\ g(x) \ne y \\
          \{\top\} & \text{if}\ x \in \mathrm{dom}(g)\ \text{and}\ g(x)=y \\
          \{\bot,\top\} & \text{if}\ (x,y) \in E\ \text{and}\ x \notin
          V_0
        \end{array}\right.
      \]
    \end{itemize}
\end{restatable}


\subsection{Reduction to the Dyck problem}
\label{subsec:reductiontodyck}



In this section, we~present the reduction of
$\Dyn\NonUnifParity_{\bfC,\kappa}$ to a Dyck-path problem on an
acyclic labelled graph. Our reduction is such that any update (of the
edges) in the parity game corresponds to a simple update of the
acyclic graph.
As~we explain in the next section, this reduction will prove
Theorem~\ref{mainresult}.  
  Our proof is built on a transformation of the
  graph~$G$ of the parity game into an acyclic
  labelled graph~$\Gamma_G$ for the Dyck-path problem.
  This construction takes advantage of the iterative computation of the selector views:
  most of the computation of $\omega_i^{g}$ depends only of $\omega_{i+1}^g$,
  and only few steps (for critical indices $i \leq \ell-1$)
  of the computation actually directly depend on~$g$ itself.


\smallskip    
    A ``nominal'' vertex of graph~$\Gamma_G$ is a pair $(i,Z)$, where $Z \in \Lambda_i$
    stores the possible values for~$\omega_i^g$ (with~$g$ ranging
    over the set of $V_0$-selectors).

    Using Proposition~\ref{prop:decremental-II}, for all non-critical values of     
    $i \leq \ell-1$, the value of $\omega_i^g$ can be computed from only
    $\omega_{i+1}^g$ (for every $g$). Hence we add ``neutral'' edges,
    labelled with~$\bullet$ (those are labels that will not need
      to be well-balanced in Dyck paths):
    \[(i+1,Z) \xrightarrow{\bullet} (i,Z') \text{ where $Z' = \Omega_i(Z)$.}\]
    Note that these transitions always exist, whatever the set of edges
    in~$G$. In~particular, they have to be precomputed, and are not affected
    by insertions or deletions of edges in~$G$.

    Now, if $i \leq \ell-1$ is critical, we have to abstract the choice of~$g$, and
    to give all possible options in the graph~$\Gamma_G$. For $x,y \in V$, let
    $\choice(x,y)$ be defined as follows:
    \[
    \choice(x,y) = \begin{cases}\{\{\bot\}\} & \text{if } (x,y) \notin E \\
      \{\{\bot\},\{\top\}\} & \text{if } (x,y) \in E \text{ and } x \in V_0 \\
      \{\{\bot,\top\}\} & \text{if } (x,y) \in E \text{ and } x \in
      V_1\end{cases}
    \]
    This represents the possible options of player~$P_0$ from state~$x$:
    in~the first case, edge~$(x,y)$ is not available; in
    the second case, edge~$(x,y)$ exists in~$G$, and player~$P_0$ can decide to
    take it or~not; in~the last case, player~$P_0$ should take both
    possibilities into account. 

    The natural way of building~$\Gamma_G$ would consist in adding neutral edges
    $(i,Z) \xrightarrow{\bullet} (i,Z')$ for all
    ``relevant'' vectors $\mathbf{B}$ and $\mathbf{B}'$, with $Z' =\Omega_i(Z,\mathbf{B},\mathbf{B}')$.
    However, this would require too many changes in~$\Gamma_G$ when an edge is updated
    in~$G$. Instead, we~split such an edge into several edges, and
    consider Dyck paths instead.
    
    Let us say that a pair of vectors $(\mathbf{B}_{\theta_i,s})_{s \in \mathbf{T}(v_i)}$
    and $ (\mathbf{B}'_{s,\theta_i})_{s \in \mathbf{T}(v_i)}$
    is \emph{suitable at step $i$} if their entries are non-empty subsets of $\{\bot,\top\}$.
    The graph $\Gamma_G$ additionally
    contains the following edges and vertices: 
    \begin{itemize}
    \item $(i+1,Z) \xrightarrow{i,Z} (i,0)$ for all $Z \in \Lambda_i$;
    \item $(i,0) \xrightarrow{i,\mathbf{B},\mathbf{B}'} (i,1)$ for all vectors $\mathbf{B}$ and $\mathbf{B}'$
    suitable at step $i$ and with entries $\mathbf{B}'_{x,y}, \mathbf{B}_{x,y} \in \choice(x,y)$;
    \item $(i,1) \xrightarrow{\overline{i,\mathbf{B},\mathbf{B}'}} (i,\mathbf{B},\mathbf{B}')$
    for all vectors $\mathbf{B}$ and $\mathbf{B}'$ suitable at step $i$;
    \item $(i,\mathbf{B},\mathbf{B}') \xrightarrow{\overline{i,Z}} (i,Z')$ where $Z' = \Omega_i(Z,\mathbf{B},\mathbf{B}')$,
    for all $Z \in \Lambda_i$ and for all vectors $\mathbf{B}$ and $\mathbf{B}'$ suitable at step $i$.
    \end{itemize}

\medskip
    Observe that every path of the form $(i+1,Z) \xrightarrow{\bullet} (i,Z')$ or
    \[(i+1,Z) \xrightarrow{i,Z} (i,0) \xrightarrow{i,\mathbf{B},\mathbf{B}'} (i,1) \xrightarrow{\overline{i,\mathbf{B},\mathbf{B}'}} 
    (i,\mathbf{B},\mathbf{B}') \xrightarrow{\overline{i,Z}} (i,Z')\]
    is a Dyck path in $\Gamma_G$: we call such a path a \emph{nominal Dyck path}.
Then, due to the layered structure of the graph $\Gamma_G$, every Dyck
path in $\Gamma_G$ is either a sub-path of a nominal Dyck path (i.e., an empty path or
a 2-step-long path of the form
$(i,0) \xrightarrow{i,\mathbf{B},\mathbf{B}'} (i,1) \xrightarrow{\overline{i,\mathbf{B},\mathbf{B}'}} 
(i,\mathbf{B},\mathbf{B}')$) or
a concatenation of nominal Dyck paths; in the latter case, we call the path a \emph{generic Dyck path}.

%
    Finally, we~notice that insertion and deletion
    of an edge~$(x,y)$ in~$G$ can be quickly propagated in~$\Gamma_G$:  
    such changes only affect
    the edges $(i,0) \xrightarrow{i,\mathbf{B},\mathbf{B}'} (i,1)$
    when $x=\theta_i$ and $y \in \mathbf{T}(v_i)$, or $y = \theta_i$ and
    $x \in \mathbf{T}(v_i)$ (where $\mathbf{B}$ and $\mathbf{B}'$ are then given by the $\choice$
    function). They can easily be performed in \FO.

\smallskip
    We focus now on the subclass $\mathcal{P}_\Gamma$ of generic Dyck paths in
    $\Gamma_G$ with source $(\ell,\Omega_\ell)$
    and with sink in $\{(1,Z) \mid Z \subseteq \Lambda(v_1)\}$.  For
    each path $\pi \in \mathcal{P}_\Gamma$, we~let
    \begin{itemize}
    \item  $Z(\pi)$ be the subset of $\Lambda(v_1)$ such that
      $(1,Z(\pi))$ is the sink of~$\pi$;
    \item $E(\pi)$ be the subset of $E$ defined as
 \[\arraycolsep=1pt\begin{array}{lll}E(\pi) & = &
   \left\{(\theta_i,s) \left|\begin{array}{l} i \text{ is critical, } \theta_i \in V_0 \text{, } s \in \mathbf{T}(v_i) \text{ and } \\ \pi \text{ goes through a vertex }
   (i,\mathbf{B},\mathbf{B}') \text{ with } \mathbf{B}_{\theta_i,s} = \{\top\}\end{array}\right\}\right. \cup \\
   & & \left\{(s,\theta_i) \left|\begin{array}{l} i \text{ is critical, } s \in V_0 \cap \mathbf{T}(v_i) \text{ and } \\ \pi \text{ goes through a vertex }
   (i,\mathbf{B},\mathbf{B}') \text{ with } \mathbf{B}'_{s,\theta_i} = \{\top\}\end{array}\right\}\right..\end{array}\]
\end{itemize}

The correctness of the reduction follows from the construction in
Section~\ref{section:decremental-II}.

\begin{proposition}
\label{prop:accessibility}
For every $V_0$-selector $g$ in $G$, there exists a unique Dyck path
$\pi^g \in \mathcal{P}_\Gamma$ such that $E(\pi^g) = \{(s,g(s)) \mid s
\in \mathrm{dom}(g)\}$.  The function $g \mapsto \pi^g$ is a
bijection from the set of $V_0$-selectors to $\mathcal{P}_\Gamma$.
Moreover, we have $Z(\pi^g) = \omega_{i}^g$, i.e. $g$ ensures
$P_0$'s victory when the play is in~$\sigma$ if, and only~if, $Z(\pi^g)
= \{\sigma \mapsto \top\}$.  
\end{proposition}

Consequently, answering our parity-game problem
amounts to checking if there exists a Dyck path $\pi \in
\mathcal{P}_\Gamma$ such that $Z(\pi) = \{\sigma \mapsto \top\}$.
A~\Dyn\FO algorithm has been proposed for this problem in~\cite{WS07}
in the case of two-letter alphabets; extending it
to alphabets of polynomial size, as is required here, is straightforward.

\section{Overall complexity analysis}
\label{section:overall}

In this section, we~analyze the complexity of our dynamic algorithm.
It~is worth noticing that our transformation above is not a
\emph{bfo-reduction} in the sense of~\cite{PI97,HI02}: the ``easy''
(edgeless) instances of the parity-game problem are not mapped to the
``easy'' instances of the Dyck-path problem. In~order to accommodate
with this, our precomputation phase also has to compute the auxiliary
information for the Dyck-path algorithm in the image of the
  edgeless game. As~we~explain in Appendix~\ref{app:complexity}, this
can be achieved in \LOGSPACE.


In the end, during the precomputation phase, 
the algorithm computes:
\begin{itemize}
\item a nice tree decomposition $\mathcal{D} = (\mathcal{T},\mathbf{T})$ of the maximal graph $(V,E_\star)$,
rooted at $\sigma$, and a depth-first traversal of $\mathcal{T}$;
\item the vertices and edges of the graph~$\Gamma_{G_\epsilon}$, where $G_\epsilon$ is the edgeless graph
  $(V,\emptyset,c,V_0,V_1)$;
\item useful
  predicates over the graph~$\Gamma_{G_\epsilon}$, in order to set up
  our dynamic algorithm for solving the Dyck
  reachability problem.
\end{itemize}

\smallskip 

During the update phases, whenever introducing or deleting an
edge~$e$ in~$G$, we~have to delete and introduce edges of the form
$(i,0) \xrightarrow{i,\mathbf{B},\mathbf{B}'} (i,1)$ in $\Gamma_G$,
for a given index $i$ that depends on $e$.
A constant (independent of $e$) number of such edges is concerned,
and, as mentioned in Section~\ref{subsec:reductiontodyck},
these edges are identified by \FO formulas taking the
edge~$e$ into account.
Moreover, since $\mathcal{D}$ is a tree decomposition of the maximal graph $(V,E_\star)$,
it remains a tree decomposition of $G$ throughout the entire computation.
Consequently, updating the edge-membership
predicate of~$\Gamma_G$ and the useful auxiliary predicates
can be done with \FO formulas 
  since the Dyck reachability problem is in
  $\Dyn\FO$.
Finally, deciding
whether~$\sigma$ is winning from these predicates can be done using
\FO formulas again, which completes the proof of
Theorem~\ref{mainresult}.

\section{Conclusion}
\label{sec-concl}

We have presented in this extended abstract a dynamic algorithm for
deciding whether player~$P_0$ is winning from a given state~$\sigma$ of
a parity game in which some edges may be added or deleted.

As~we explain in Appendices~\ref{app:dyck} to~\ref{app-extensions}, our algorithm can
readily be extended in several directions:
\begin{itemize}
\item \emph{non-uniform strategy synthesis}: besides deciding
  whether a given state is winning for player~$P_0$, the algorithm also
  maintains a predicate characterizing a winning strategy,
  should one exist;
\item \emph{uniform winning strategy}: our algorithm
  also computes the entire set of winning states for player $P_0$, as~well as
  a uniform memoryless winning strategy (i.e., a~single memoryless
  strategy that is winning from all the winning states);
%
\item \emph{additional updates}: our algorithm is extended to also
  handle modifications of the color of a state (within the fixed
  range~$[1,\bfC]$), or the change of owner of a state. 
\end{itemize}
%

\smallskip

As a final remark, let~us mention that our algorithm can be used for
dynamic model checking of a fixed $\mu$-calculus formula over Kripke structures
of bounded tree-width: $\mu$-calculus model checking can be expressed as
a parity game~\cite{EJ91,BW16}, by considering the product of the structure with the
graph of the $\mu$-calculus formula. Assuming that the formula is fixed
and that we~have a tree decomposition of the input structure, we~can
build the parity game and a tree decomposition for this game in \LOGSPACE. Moreover,
any addition and deletion of an edge in the Kripke structure translates
to the addition or deletion of a fixed number of edges in the parity
game. It~follows that dynamic model checking of a fixed $\mu$-calculus
formula over bounded-tree-width Kripke structures is in
\Dyn[\logspace]\FO. This can be compared to the $\Dyn{\TC[0]}$ algorithm
of~\cite{KW03} for dynamic model checking of fixed \LTL formulas over
arbitrary Kripke structures.



%



\label{finalpage}


\makeatletter
\newcounter{suitebib}
\let\@save@bibitem\@bibitem
\def\@bibitem#1{\@save@bibitem{#1}
  \global\setcounter{suitebib}{\value{\@listctr}}}
\makeatother

\putbib[nmbib]
\end{bibunit}


\ifnoapp\else
\clearpage
\appendix
\begin{bibunit}


\centerline{\noindent{\LARGE\bfseries --- Technical Appendix ---}}

\bigskip\bigskip


This appendix contains proofs and extensions of the results presented in
the main part of our paper:
\begin{itemize}
\item Appendix~\ref{app-dynfo} contains a more precise exposition of
  dynamic complexity;
\item Appendix~\ref{app-treedec} contains the proofs of the results of
  Section~\ref{sec-treedec}, concerned with tree decompositions and
  linear decompositions;
\item Appendix~\ref{app-stratviewI} contains the additional material and
  proofs about selector views;
\item Appendix~\ref{app-overall} gives a detailed account of the
  overall complexity of our algorithm, explaining why it is in
  \Dyn[\LOGSPACE]\FO;
\item Appendix~\ref{app-dyck} extends the algorithm of~\cite{WS07} to
  also maintain the auxiliary information for computing a witnessing
  Dyck path, should one exist;
\item Appendix~\ref{app-access} defines and studies \emph{accessibility
    sets}, which will be useful for computing the winning region of
  parity games;
\item Appendix~\ref{app-gammaG} contains a refined construction of the
  graph $\Gamma_G$, which we query for Dyck paths. The construction uses
  accessibility sets, in order to be able to compute the winning region
  and a uniform winning strategy;
\item Appendix~\ref{app-overall-2} extends the complexity results to the
  framework of Appendices~\ref{app-dyck} to~\ref{app-gammaG};
\item Appendix~\ref{app-extensions} explains how our algorithm can be
  extended to solve the uniform version of our problem, and to take
  additional update operations into account.
\end{itemize}

\section{Dynamic complexity classes}
\label{app-dynfo}

In this section, we~briefly introduce the formalism of descriptive
complexity, and extension to dynamic complexity. We~refer to
e.g.~\cite{Imm99,PI97,Hes03,Zeu15} for more details.




\smallskip
A~\emph{vocabulary} is a tuple $\tau=\tuple{R_i^{a_i},c_j}$ of relations
and constant symbols. A~\emph{$\tau$-structure} is a tuple
$\calS=\tuple{|\calS|, R_i, c_j}$ where $|\calS|=[0;n-1]$ is the
universe of size~$n\geq 1$, $R_i\subseteq |\calS|^{a_i}$ are relations,
and $c_j\in |\calS|$ are constants. We~write $\mathsf{Struct}_n(\tau)$ for the set
of finite structures of size~$n$ over vocabulary~$\tau$.
Any~subset~$S\subseteq \mathsf{Struct}_n(\tau)$ defines a static decision problem,
where the aim is to decide whether a given structure in $\mathsf{Struct}_n(\tau)$
belongs to~$S$. Descriptive complexity aims at characterizing such
subsets using logical formulas. As~an example, a~graph can be
represented with its set of states as universe and its edges as a binary
relation; that each state has at least one outgoing transition can be
expressed as the first-order formula $\forall x.\ \exists y.\ E(x,y)$.
The~class~\FO of problems that can be characterized by first-order
formulas is known to correspond to the class~$\AC[0]$ in circuit
complexity~\cite{BIS90}. 

\smallskip
In order to define dynamic-complexity problems, we~have to define
\emph{operations} on structures. Given an operation symbol~$\sigma$ (from a finite
alphabet~$\Sigma$) with arity~$m$, an integer~$n$, and $m$
elements~$a_1$, $a_2$,~...,~$a_m$ of~$|\calS|=[0;n-1]$,
the~operation~$\sigma(a_1,...,a_m)$ maps structures of size~$n$ to structures
of size~$n$. We~write~$\Sigma_n$ for the set of those operations.
The~semantics of those operations is defined as update functions~$g$ mapping a
structure~$A\in\mathsf{Struct}_n(\tau)$ and an operation
in~$w\in\Sigma_n$ to a structure in~$\mathsf{Struct}_n(\tau)$, which
we denote~$w(A)$. For a finite sequence of
operations~$W=(w_i)_{1\leq i\leq k}$ of operations in~$\Sigma_n$, we~write
$W(A)$ for the structure $w_k(w_{k-1}(\cdots(w_1(A))\cdots))$.

We further need to define \emph{initializations} on structures.
  Given a vocabulary~$\tau$ and an initialization vocabulary~$\tau^{\mathrm{init}}$,
  an initialization is a function $f$ mapping a structure
  $A \in \mathsf{Struct}_n(\tau^{\mathrm{init}})$ to a structure in~$\mathsf{Struct}_n(\tau)$,
  which we denote by $\calS_A^n$.
  Hence, we may view the set $\mathsf{Struct}_n(\tau^{\mathrm{init}})$ as an
  initialization alphabet $\Xi_n$.

  Now, given a vocabulary~$\tau$, a static decision problem~$S$
  over~$\tau$, an initialisation alphabet $\Xi = \bigcup_n
      \Xi_n$, an initialisation function $f$ and a set operation
  symbols~$\Sigma = \bigcup_n \Sigma_n$, the output of
  the~\emph{dynamic problem} associated with~$S$ and~$\Sigma$ is the
  set of non-empty words~$w$ such that $w\in\Xi_n \cdot \Sigma_n^\ast$
  for some~$n\geq 1$, and $w(\calS_A^n)\in S$.

As a special case, which we are concerned with in this article,
the initialization may consist in restricting
the set of allowed operations to a set $\Sigma_A \subseteq \Sigma_n$.
This case may be achieved when the vocabulary $\tau^{\mathrm{init}}$ consists only
in the relations $\sigma \in \Sigma$, i.e. when $\Xi_n = 2^{\Sigma_n}$
and when $\tau$ is the union of $\tau^{\mathrm{init}}$ with another vocabulary $\tau'$.
Then, the structure $\calS_A^n$ associated with a structure $A \in \Xi_n$
may simply be seen as a tuple $\tuple{A,\calS_0^n}$
where $\calS_0^n$ is a reference structure over the vocabulary $\tau'$
(e.g. the empty structure).
%

\smallskip

A~dynamic algorithm is a uniform family $(M_n)_{n\in\bbN}$ of
deterministic finite-state automata
$M_n=\tuple{Q_n,\Xi_n,\Sigma_n,\delta_n^{\mathrm{init}},\delta_n^{\mathrm{up}},s_n,F_n}$,
where $Q_n$ is the finite
set of $\tau^{\mathrm{aux}}$-structures of size~$n$ (where $\tau^{\mathrm{aux}}$ is the vocabulary
of \emph{auxiliary information}), $\Xi_n$~is the set of initializations, $\Sigma_n$~is the set of
update operations, $\delta_n^{\mathrm{init}}$ is the deterministic \emph{initial} transition function,
$\delta_n^{\mathrm{up}}$ is the deterministic \emph{update} transition function,
$s_n$~is the initial state and $F_n$~is the set of final states.
The initial transition function is used for reading the initial letter (chosen from $\Xi_n$) only,
and the update transition function is used for reading all subsequent letters (chosen from $\Sigma_n$).

For two
static complexity classes~$\calC$ and~$\calC'$, a~dynamic algorithm
belongs to $\Dyn[\calC']{\calC}$ (written $\Dyn\calC$ when
$\calC'=\calC$) when
\begin{itemize}
\item the update transition function can be computed in~$\calC$ (more
  precisely, the values of the predicates in the new state can be
  computed in~$\calC$ from the values in the previous state and from the
  operation being performed);
\item whether a state is in~$F_n$ can be decided in~$\calC$;
\item the initial state and the initial transition function
can be computed in~$\calC'$.
\end{itemize}
In~the sequel, we only consider the case where $\calC=\FO$, meaning
that first-order formulas will be used to describe how predicates are
updated along transitions.

\section{Additional material for Section~\ref{sec-treedec}}
\label{app-treedec}

\label{app-pureTD}

\puredecomposition*

\begin{proof}
%
  It is known~\cite{EJT10} that we can contruct in \LOGSPACE~a
  tree decomposition $\mathcal{D} = (\mathcal{T},\mathbf{T})$
  of $G = (V,E)$ with width at most $4 \kappa+3$ and
  with diameter at most $c'(\kappa) (\log_2(n)+1)$, for some integer constant $c'(\kappa)$.
  It remains to transform $\mathcal{D}$ into a nice tree decomposition.
  
  First, for all vertices $s \in V$, we select a node $v$ of $\mathcal{T}$ such that
  $\sigma \in \mathbf{T}(v)$; we add a new node $v_s$ and a new edge $(v_s,v)$
  to $\mathcal{T}$, and we set $\mathbf{T}(v_s)$.
  
  Second, for every edge $(v,w)$ of $\mathcal{T}$,
  let us write $\mathbf{T}(v) = \{s_1,\ldots,s_k\}$ and $\mathbf{T}(w) = \{t_1,\ldots,t_\ell\}$,
  and let us assume that we have $\mathbf{T}(v) \cap \mathbf{T}(w) =
  \{s_1,\ldots,s_m\} = \{t_1,\ldots,t_m\}$ for some $m$.
  If $k = \ell = m$, then we just merge both nodes.
  
  Otherwise, without loss of generality we assume that $m < \ell$.
  We add further nodes $v_m,v_{m+1},\ldots,v_{k-1}$ and $w_{m+1},\ldots,w_{\ell-1}$
  and edges $(v_m,v_{m+1}),\ldots,(v_{k-2},v_{k-1})$, $(v_{k-1},v)$,
  $(v_m,w_{m+1}),\ldots,(w_{\ell-2},w_{\ell-1})$, $(w_{\ell-1},w)$ to $\mathcal{T}$,
  and we set $\mathbf{T}(v_i) = \{s_1,\ldots,s_i\}$ and $\mathbf{T}(w_i) = \{t_1,\ldots,t_i\}$.
  
  Finally, we repeat the following two operation whenever available:
  if there exists a leaf $v$ such that $v \notin \{v_s \mid s \in V\}$,
  then we delete the node $v$ and its (only) incident edge.
  
  Performing all these transformations can be done in \LOGSPACE~and
  leaves us with a nice tree decomposition of width at most $4 \kappa+3$
  and with diameter at most $c(\kappa)(\log_2(n)+1)$, where
  $c(\kappa) = 8(\kappa+1)(c'(\kappa)+2)$.
\end{proof}

\section{Additional material for Subsections~\ref{sec-refselview}
  and~\ref{section:decremental-II}}
\label{app-stratviewI}

\subsection{Decremental computation of selector views: Proof of
  Proposition~\ref{prop:decremental-I}}
\label{section:decremental-I}

\propdecrI*


\begin{proof}
Let $f$ be a $V$-selector. We compute
selector views~$f_i$ and in a \emph{decremental} way.

First, since $|V| \geq 2$, the nice tree decomposition $(\mathcal{T},\mathbf{T})$
must contain at least two nodes, hence $\ell \geq 2$.
It follows that $\Psi_{\ell} = \Psi_1 \subseteq \Psi_{\leq \ell-1}$,
which proves that $\ell$ is not critical and that $\Psi_{>\ell-1} = \emptyset$.
Therefore, we have $\Psi_{>\ell-1}(v) = \emptyset$ for all $v \in \Theta_\ell$,
whence $f_\ell(v,s) = ?$ for all $s \in \mathbf{T}(v)$.

We assume now that $i \leq \ell-1$ and we wish to compute the selector view $f_i$ using~$f_{i+1}$.
The value of $f_i(v)$ depends on the set $\Psi_{>i-1}(v) = \{s \in \Psi_{>i-1} \mid n_i(s) = v\}$.
Hence, in order to compute $f_i(v)$, we must take into account two facts:
(a) for some vertices $s \in V$, the nodes $n_i(s)$ and $n_{i+1}(s)$ may differ
(if $v_{i+1}$ is a child of $v_i$), and
(b) the sets $\Psi_{>i-1}$ and $\Psi_{>i}$ may differ (if $i$ is critical).
We first deal with the former fact,
and will tackle the latter one afterwards.

For all nodes $v \in \Theta_i$, we denote by
$\Theta^\ast_{>i}(v)$, $\Psi^\ast_{>i}(v)$ the sets $\{w \in \Theta_{>i}
\mid n_i(w) = v\}$ and $\{s \in \Psi_{>i} \mid n_i(s) = v\}$.
Every edge of~$G$ with
one end~$s$ in~$\Psi_{>i}^\ast(v)$ has its other end in $\Psi_{>i}^\ast(v)
\cup \bfT(v)$, since $\Theta_{>i}^\ast(v)\cup \{v\}$ is a subtree of
$\mathcal{T}$ that contains $\calT^s$.

Then, we introduce the following variant of the selector views.
Consider an integer $1\leq i \leq \ell$, a node $v \in \Theta_i$,
a vertex $s \in \mathbf{T}(v)$ and a $V$-selector~$f$. Let $\pi = s \cdot
s_1 \cdot s_2 \dots$ be the unique (maximal) outcome of~$f$ from~$s$.
We denote by $f_i^\ast(v, s)$, the element
of $\{?,\top,\bot\} \cup (\bfT(v) \times \{1,\ldots,\mathbf{C}\})$
defined as follows:
  \begin{itemize}
  \item if $\pi$ stays in~$\Psi_{>i}^\ast(v)$ after the first step (or if
    $s \in \Psi_{>i}^\ast(v) \setminus \mathrm{dom}(f)$) and is winning
    for~$P_0$, then $f_i^\ast(v,s) = \top$;
  \item if $\pi$ stays in~$\Psi_{>i}^\ast(v)$ after the first step (or if
    $s \in \Psi_{>i}^\ast(v) \setminus \mathrm{dom}(f)$) and is losing
    for~$P_0$, then $f_i^\ast(v,s) = \bot$;
  \item if the first step of $\pi$ does not enter $\Psi_{>i}^\ast(v)$ (or
    if $s \notin \Psi_{>i}^\ast(v)$ and $s \notin \mathrm{dom}(f)$), then
    $f_i(v,s) = \mathord?$;
  \item if $\pi$ stays in~$\Psi_{>i}^\ast(v)$ from the first step until
    some step~$j\geq 2$, and $c_m$~is the maximal color
    seen along~$\pi$ between~$s$ and~$s_{j-1}$, then $f_i^\ast(v,s) =
    (s_j,c_m)$.
  \end{itemize}
Note this definition is the same as Definition~\ref{dfn:selector-view-f},
where we use the set $\Psi_{>i}^\ast(v)$ instead of $\Psi_{>i-1}(v)$.
We also denote by $f_i^\ast(v)$ the function $s \mapsto f_i^\ast(v,s)$.

Now, we compute $f_i^\ast$ using $f_{i+1}$.
\begin{itemize}
\item If $v_{i+1}$ is the parent of $v_i$, then
  $\Theta_{i+1} = \mathcal{A}(v_{i+1}) \subseteq \mathcal{A}(v_i) = \Theta_i$,
  whence $\Theta_{\leq i} = \Theta_{\leq i+1}$, and thus
  $n_i(s) = n_{i+1}(s)$ for all vertices $s \in \Psi_{>i}$.
  It follows that $\Psi_{>i}^\ast(v_i)$ is empty (since $v_i \notin \Theta_{i+1}$)
  and that $\Psi_{>i}^\ast(w) = \Psi_{>i}(w)$ for all $w \in \Theta_{i+1} = \Theta_i \setminus \{v_i\}$.
  
  Hence, we have $f_i^\ast(v_i) : s \mapsto \mathord?$ and $f_i^\ast(w) = f_i(x)$
  for all nodes $w \in \Theta_i \setminus \{v_i\}$.
 
\item If $v_{i+1}$ is a child of $v_i$, then
  $v_{i+1} \notin \Theta_i$, hence $v_{i+1} \notin \Theta_{\leq i}$ and
  $\Theta_{\leq i+1} = \Theta_{\leq i} \uplus \{v_{i+1}\}$.
  Consequently, we have $n_i(s) = v_i$ for all vertices $s \in \Psi_{>i}(v_{i+1})$, and
  $n_i(s) = n_{i+1}(s)$ for all other vertices $s \in \Psi_{>i}$.
  It follows that $\Psi_{>i}^\ast(w) = \Psi_{>i}(w)$ (and thus that $f_i^\ast(w) = f_{i+1}(w)$)
  for all nodes $w \in \Theta_i \setminus \{v_i\}$, and that 
  $\Psi_{>i}^\ast(v_i) = \Psi_{>i}(v_i) \uplus \Psi_{>i}(v_{i+1})$. 
  
  Then, let $\Lambda(v_{i+1},v_i)$ be the set of functions
  $\varphi \in \Lambda(v_{i+1})$ such that $\varphi(s) \in \{\mathord?,\top,\bot\} \cup
  (\mathbf{T}(v_i) \times \{1,\ldots,\mathbf{C}\})$ for all $s \in \mathbf{T}(v_{i+1})$.
%
  If $i+1$ is critical, then $\theta_{i+1} \in \Psi_{>i}(v_{i+1})$, hence
  $\mathbf{T}(v_{i+1}) = \mathbf{T}(v_i) \uplus \{\theta_{i+1}\}$, and
  we have $f_{i+1}(v_{i+1}) \in \Lambda(v_{i+1},v_i)$.
  If $i+1$ is non-critical, then $\mathbf{T}(v_{i+1}) \subseteq \mathbf{T}(v_i)$,
  and we also have $f_{i+1}(v_{i+1}) \in \Lambda(v_{i+1},v_i)$.
  
  It follows easily that $f_i(v_i) = \pi_i^1(f_{i+1}(v_i),f_{i+1}(v_{i+1}))$, where the
  function $\pi_i^1 \colon \Lambda(v_i) \times \Lambda(v_{i+1},v_i) \to \Lambda(v_i)$ is defined
  by:
  \[
  \pi_i^1(\varphi_1,\varphi_2) \colon s \mapsto \begin{cases} \varphi_2(s) &
    \text{if } \varphi_2(s) \neq \mathord ? \\
    \varphi_1(s) & \text{otherwise}.\end{cases}
  \]
\end{itemize}

\medskip

Finally, we compute $f_i$ when $i \leq \ell-1$. If $i$ is non-critical, then $\Psi_{>i-1} = \Psi_{>i}$,
and therefore $f_i = f_i^\ast$. However, if $i$ is critical,
then computing the selector view $f_i$ requires using~$f_i^\ast$
and the sets $S = \{s \in \mathbf{T}(v_i) \mid f(\theta_i) = s\}$ and $T = \{s \in \mathbf{T}(v_i) \mid f(s) = \theta_i\}$.

First, observe that $\theta_i$ is the only element of $\Psi_{>i-1} \setminus \Psi_{>i}$.
Since $n_i(\theta_i) = v_i$, it follows that, for all
nodes $w \in \Theta_i \setminus \{v_i\}$, we have
$\Psi_{>i-1}(w) = \Psi_{>i}^\ast(w)$ and thus $f_i(w) = f_i^\ast(w)$.
Hence, it remains to compute $f_i(v_i)$.

In addition, observe that $S$ is either empty or a singleton.
Then, we compute $f_i(v_i)$ in three steps:
\begin{itemize}
 \item We first take $f_i^\ast$, $S$ and $T$ into account into describing the outcome of $f$ from vertices $s \in \mathbf{T}(v_i)$.
    More precisely, we define a fuction $\pi^2_i \colon \Lambda(v_i) \times 2^{\mathbf{T}(v_i)} \times 2^{\mathbf{T}(v_i)} \mapsto \Lambda(v_i)$ by:
    \[
      \pi^2_i(\varphi,S,T) : x \mapsto \begin{cases} (s,c(\theta_i)) & \text{if } x = \theta_i \text{ and if } S \text{ is a singleton set } \{s\} \\
      (\theta_i,c(x)) & \text{if } x \in T \\
      \varphi(x) & \text{otherwise.}\end{cases}
    \]
    Note that, $\theta_i \in T$ if, and only if, $S = \{\theta_i\}$, so that both definitions of $\pi^2_i(\varphi,S,T)(\theta_i)$ coincide.
 \item Then, we compute $f_i(v_i,\theta_i)$ using $\pi^2_i(f_i^\ast(v_i),S,T)$:
    it comes directly that $f_i(v_i,\theta_i) = \pi^3_i(\pi^2_i(f_i^\ast(v_i),S,T))$, where the function
    $\pi^3_i \colon \Lambda(v_i) \to \{?,\top,\bot\} \cup (\bfT(v) \times \{1,\ldots,\mathbf{C}\})$
    is defined by:
    \[
      \pi^3_i(\varphi) = \begin{cases} \top & \text{if } \varphi(\theta_i) = ? \text{ and } \theta_i \in V_1 \text{, or if } \varphi(\theta_i) \in \{\theta_i\} \times (2\mathbb{N}) \\
      \bot & \text{if } \varphi(\theta_i) = ? \text{ and } \theta_i \in V_0 \text{, or if } \varphi(\theta_i) \in \{\theta_i\} \times (2\mathbb{N}+1) \\
      \varphi(x) & \text{otherwise.}\end{cases}
    \]
 \item Finally, we compute $f_i(v_i)$ using $\pi^2_i(f_i^\ast(v_i),S,T)$ and its image by $\pi^3_i$:
     again, if comes at once that $f_i(v_i) = \pi^4_i(\pi^2_i(f_i^\ast(v_i),S,T))$, where the function
     $\pi^4_i \colon \Lambda(v_i) \mapsto \Lambda(v_i)$ is defined by:
     \[
       \pi^4_i(\varphi) : x \mapsto \begin{cases} \pi^3_i(\varphi) & \text{if } x = \theta_i \text{, or if }
       \varphi(x) \in \{\theta_i\} \times \mathbb{N} \text{ and } \pi^3_i(\varphi) \in \{\bot,\top\}\\
       (s,\max\{c_1,c_2\}) & \text{if } \varphi(x) = (\theta_i,c_1) \text{ and } \pi^3_i(\varphi) = (s,c_2) \text{ for some } s \text{, } c_1 \text{ and } c_2\\
       \varphi(x) & \text{otherwise.}\end{cases}
     \]
\end{itemize}

The computations performed above were all \FO-definable.
More precisely, $f_i^\ast$ is always \FO-definable using $f_{i+1}$ for $i\leq \ell-1$;
furthermore, $f_i$ is \FO-definable using $f_i^\ast$ if $i$ is non-critial,
or using $f_i^\ast$ and the sets $S$ and $T$ if $i$ is critical.
This completes the proof.
\end{proof}

\subsection{Decremental computation of views of \texorpdfstring{\(V_0\)}{V0}-selectors:
  Proof of Proposition~\ref{prop:decremental-II}}
\label{app:decremental-II}


\lemvictory*

\begin{proof}
  A $V_0$-selector~$g$ is winning for $P_0$ at $\sigma$ if all its
  extensions lead to $P_0$'s victory when starting from~$\sigma$, i.e.
  if $f_{1}(v_1,\sigma) = \top$ for all extensions~$f$ of~$g$ or,
  equivalently, if $\omega_{1}^g(v_1) = \{\sigma 
  \mapsto \top\}$.  Lemma~\ref{lem:victory} follows from Lemma~\ref{lem-wintop}.
\end{proof}

\lemsoundness*

\begin{proof}
Consider a subset $X$ of $\Theta_i$ and a tuple $(\varphi_v)_{v \in X} \in \mathcal{C}^0 \cap
\prod_{v \in X} \omega_{i}^g(v)$. Let $(f^v)_{v \in X}$ be
extensions of~$g$ such that $\varphi_v = f^v_{i}(v)$ for all $v \in
X$. We construct an extension~$f$ of~$g$ in three successive steps, as
follows:
\begin{enumerate}
\item for all nodes $v \in X$ and for all vertices $s \in \Psi_{>i-1}(v)
  \cup \{t \in \mathbf{T}(v) \mid \varphi_v(t) \neq \mathord?\}$,
  we set $f(s) = f^v(s)$;
\item for all vertices $s$ not treated above, we set $f(s) = g(s)$,
\end{enumerate}
with the convention that, if a vertex $s$ lies outside of the domain of a selector~$\varphi$, then
setting $f(s) = \varphi(s)$ amounts to excluding $s$ from the domain of $f$.
It is then straightforward to check that the $V$-selector $f$ is well-defined and extends $g$
and that $\varphi_v = f^v_{i}(v)$ for all $v \in X$.
%
\end{proof}

\propdecrII*

\begin{proof}
  Thanks to Lemma~\ref{lem:soundness}, the above proof of Proposition~\ref{prop:decremental-I}
  leads to the following decremental recursive characterization of $\omega_{i}^g$.
  
  We started by showing that $\omega_\ell^g(v)$ is the singleton set
  $\{s \mapsto ?\}$, for all nodes $v \in \Theta_\ell$ and all $V_0$-selectors $g$.
  Hence, $\omega_\ell^g$ is \FO-definable and does not depend on $g$.
  
  Then, we follow the two-step approach used while proving Proposition~\ref{prop:decremental-I},
  and we introduce sets $\omega_i^{\ast g}(v) = \{f_i^\ast(v) \mid f \text{ is an extension of } g\}$.
  Lemma~\ref{lem:soundness} also applies to sets $\omega_i^{\ast g}$:
  for every $V_0$-selector $g$, every set $X \subseteq \Theta_i$ and every tuple $(\varphi_v)_{v \in X} \in
  \mathcal{C}^0 \cap \prod_{v \in X}\omega_i^{\ast g}(v)$, there exists an extension $f$ of $g$ such that
  $\varphi_v = f_i^\ast(v)$ for all nodes $v \in X$.
  
  Hence, we compute the selector view $\omega_i^{\ast g}$ as follows for $i \leq \ell-1$.
  We have $\omega_i^{\ast g}(x) = \omega_{i+1}^g(w)$ for all nodes
  $w \in \Theta_i \setminus \{v_i\}$ and, furthermore:
  \begin{itemize}
   \item if $v_{i+1}$ is the parent of $v_i$, then $\omega_i^{\ast g}(v_i) = \{s \mapsto \mathord?\}$;
   \item if $v_{i+1}$ is a child of $v_i$, then $X_{v_{i+1}} \subseteq \Lambda(v_{i+1},v_i)$, and therefore
   $\omega_i^{\ast g}(v_i) = \{\pi^1_i(\varphi_{v_i},\varphi_{v_{i+1}}) \mid (\varphi_{v_i},\varphi_{v_{i+1}})
   \in \mathcal{C}^0 \cap (X_{v_i} \times X_{v_{i+1}})\}$.
  \end{itemize}
  
  Finally, let us compute the selector view $\omega_i^g$.
  If $i$ is non-critical, then we showed that $\omega_i^g = \omega_i^{\ast g}$
  while proving Proposition~\ref{prop:decremental-I}.
  Hence, we assume that $i$ is critical.
  Again, we have $\omega_i^g(w) = \omega_i^{\ast g}(w)$ for all nodes $w \in \Theta \setminus \{v_i\}$,
  and it remains to compute $\omega_i^g(v_i)$.

  Consider the vectors $\mathbf{B} = (\mathbf{B}_{\theta_i,s})_{s \in \mathbf{T}(v_i)}$
  and $\mathbf{B}' = (\mathbf{B}'_{s,\theta_i})_{s \in \mathbf{T}(v_i)}$ whose entries are defined by:
      \[
      \mathbf{B}'_{x,y} = \mathbf{B}_{x,y} = \left\{\begin{array}{ll} \{\bot\} &
      \text{if}\
      (x,y) \notin E,\ \text{or}\ x \in V_0 \setminus \mathrm{dom}(g),\\
      & \phantom{\text{if}\ (x,y) \notin E,\ {}}
      \text{or}\ x \in
          \mathrm{dom}(g)\ \text{and}\ g(x) \ne y \\
          \{\top\} & \text{if}\ x \in \mathrm{dom}(g)\ \text{and}\ g(x)=y \\
          \{\bot,\top\} & \text{if}\ (x,y) \in E\ \text{and}\ x \notin
          V_0
        \end{array}\right.
      \]
      
  We say that a triple $(f,S,T) \in \omega_i^{\ast g}(v_i) \times 2^{\mathbf{T}(v_i)} \times 2^{\mathbf{T}(v_i)}$ is
  \emph{adapted to $g$ at step $i$} if it satisfies the following criteria:
  \begin{itemize}
   \item $S$ is of cardinality at most one, and $S$ is empty if $f(\theta_i) \neq \mathord?$;
   \item $S = \{\theta_i\}$ if and only if $\theta_i \in T$;
   \item $\top \in \mathbf{B}_{\theta_i,s}$ for all vertices $s \in S$, and
   $\bot \in \mathbf{B}_{\theta_i,s}$ for all vertices $s \in \mathbf{T}(v_i) \setminus S$;
   \item $\top \in \mathbf{B}'_{s,\theta_i}$ and $f(s) = ?$ for all vertices $s \in T$, and
   $\bot \in \mathbf{B}'_{s,\theta_i}$ for all vertices $s \in \mathbf{T}(v_i) \setminus T$.
  \end{itemize}
  
  It follows directly from Proposition~\ref{prop:decremental-I} and from Lemma~\ref{lem:soundness} that
  $\omega_i^g(v_i) = \{\pi^4_i(\pi^2_i(\varphi,S,T)) \mid (\varphi,S,T) \text{ is adapted to } g \text{ at step } i\}$.

The computations performed above were all \FO-definable, which completes the proof.
\end{proof}

\section{Complements to Section~\ref{section:overall}: Complexity
  analysis}
\label{app:complexity}\label{app-overall}
The results presented in Section~\ref{sec-stratview} (and complements
in the previous appendices) provide us with a dynamic algorithm for
solving Problem $\Dyn\NonUnifParity_{\mathbf{C},\kappa}$, which
consists in a precomputation phase 
and update phases.

During the precomputation phase, we first need to compute:
\begin{itemize}
\item a nice tree decomposition $\mathcal{D} = (\mathcal{T},\mathbf{T})$ rooted at $\sigma$
  and a depth-first traversal of $\mathcal{T}$;
\item the vertices and edges of the graph
  $\Gamma_{G_\epsilon}$ on which we will perform Dyck reachability
  tests, when $G_\epsilon$ is the graph $(V,\emptyset,c,V_0,V_1)$
  with empty set of edges;
\item the predicates $T$ and $\pi_2$ on the graph $\Gamma_{G_\epsilon}$,
\end{itemize}
where the 2-ary predicate $T$ and the 4-ary predicate $\pi_2$ are defined in~\cite{WS07} as follows.

The predicate $T(\mathbf{v}_1,\mathbf{v}_2)$
holds if, and only if, there exists a path from $\mathbf{v}_1$ to
$\mathbf{v}_2$ in $\Gamma_{\G_\epsilon}$.
The predicate $\pi_2(\mathbf{v}_1,\mathbf{v}_2,\mathbf{v}'_1,\mathbf{v}'_2)$
holds if, and only if, there exist a path $\pi$ from $\mathbf{v}_1$ to
$\mathbf{v}_2$, with label $\lambda$, and a path $\pi'$ from $\mathbf{v}'_1$ to
$\mathbf{v}'_2$, with label $\lambda'$, such that $\lambda \cdot \lambda'$ is a Dyck word.

We claim that such precomputations can be performed in \logspace.

\smallskip
First, Lemma~\ref{lemma:puredecomposition} proves that
the tree decomposition $\mathcal{D}$ can be
computed in \logspace, and has diameter at most $c(\kappa)(\log_2(n)+1)$.

\smallskip
A second issue is to compute the graph $\Gamma_{G_\epsilon}$.
We first bound the size of the tree decomposition $\mathcal{D}$,
the length $\ell$ of the depth-first traversal $v_1,\ldots,v_\ell$
and the size of the sets $\Lambda_i = \prod_{v \in \Theta_i} 2^{\Lambda(v)}$ for all $i \leq \ell$.

The tree $\mathcal{T}$ has at most $n$ leaves (since $\mathcal{D}$ is a nice tree decomposition)
and has diameter at most $c(\kappa)(\log_2(n)+1)$,
hence it contains at most $n c(\kappa)(\log_2(n)+1)$ nodes.
Then, each pair $(v_i,v_{i+1})$ of consecutive vertices of the depth-first traversal $v_1,\ldots,v_\ell$
is an oriented edge of $\mathcal{T}$, and no edge may appear twice, hence
$\ell \leq 2 n c(\kappa)(\log_2(n)+1) - 1 \leq 2n^2 c(\kappa)$.

Finally, for all nodes $v$ of $\mathcal{T}$, recall that
$|\Lambda(v)| \leq \mathbf{K}$, where $\mathbf{K} = (3 + 4(\kappa+1)\mathbf{C})^{4(\kappa+1)}$, and
that $\Theta_i$ is of size at most $c(\kappa) (\log_2(n)+1)$.
It follows that $\Lambda_i$ is of size at most $\mathbf{K}^{c(\kappa) (\log_2(n)+1)} \leq K^{c(\kappa)} n^{c(\kappa) \log_2(\mathbf{K})}$.

Consequently, we compute $\Gamma_{\G_\epsilon}$ as follows.
First, we choose some (arbitrary) linear order $\leq$ on $V$
and we perform a depth-first traversal of $\mathcal{T}$.
Then, we identify the set $[1;\ell]$ with a subset of
$V^2 \times [1;2 c(\kappa)]$, and we identify each set $\Lambda_i$
with a subset of $V^{\lceil c(\kappa) \log_2(\mathbf{K}) \rceil} \times [1;K^{c(\kappa)}]$.
Such identifications can be performed in \logspace,
for instance by deriving a linear order on each set $\Lambda_i$ from the order $\leq$.

Moreover, since each function $\Omega_i$ or $\Omega_i^\ast$ is
\FO-definable, we can decide in \logspace\ whether, for each tuple
$(\mathbf{v}_1,\mathbf{v}_2,\lambda)$ formed of two vertices and
one label of~$\Gamma_{G_\epsilon}$,
there exists an edge $(\mathbf{v}_1,\mathbf{v}_2)$ with label $\lambda$ in~$\Gamma_{G_\epsilon}$. 

\smallskip
Then, we precompute the predicate $T$, which indicates whether a vertex $\mathbf{v_1}$ is accesible from
another vertex $\mathbf{v}_2$ (regardless of labels).
Let us begin with preliminary remarks, that follow directly
from the construction of the graph $\Gamma_G$.

First, all paths of length $3$ contain at least one nominal vertex $(i,Z)$.
Henceforth, we say that a path is \emph{short} if its length is at most $3$.

Second, every nominal vertex $\mathbf{v}$ has at most one successor $\mathbf{w}$,
where a vertex $\mathbf{w}$ is said to be a \emph{successor} of a vertex $\mathbf{v}$
if there exists a labelled edge $\mathbf{v} \to \mathbf{w}$.
More precisely:
 \begin{itemize}
  \item if $\mathbf{v}$ is of the form $(i+1,Z)$ for some critical index $i$, then
    $\mathbf{w}$ exists and is equal to $(i,0)$;
  \item if $\mathbf{v}$ is of the form $(i+1,Z)$ for some non-critical index $i$, then
    $\mathbf{w}$ is of the form $(i,Z')$ if it exists.
 \end{itemize}

Consequently, for all nominal vertices $\mathbf{v}_1$ and
$\mathbf{v}_2$, checking whether there exists a path from $\mathbf{v}_1$ to $\mathbf{v}_2$
that does not use any vertex $(i,0)$ as an intermediary step is an easy task:
it suffices to start from the vertex $v = \mathbf{v}_1$,
and then to recursively replace $v$ by its unique successor,
until we either reach $\mathbf{v}_2$, or a vertex $(i,0)$, or a vertex $v$ with no outgoing edge.
Moreover, the same technique also allows us to compute the greatest integer $I$ such that
a vertex $(I,0)$ is reachable from $\mathbf{v}_1$, if any such integer exists.

Hence, consider two vertices $\mathbf{v}_1$ and $\mathbf{v}_2$ of $\Gamma_{G_\epsilon}$.
We first assume that $\mathbf{v}_1$ and $\mathbf{v}_2$ are nominal vertices, of the form
$(i_1,Z_1)$ and $(i_2,Z_2)$.
If there exists no critical index $I$ with $i_1 < I \leq i_2$,
then all paths from $\mathbf{v}_1$ to $\mathbf{v}_2$ avoid vertices of the form $(i,0)$,
thus one can check whether such paths exist.

Therefore, we also assume that there exists a critical index $I$ with $i_1 < I \leq i_2$.
Let $I_{\min}$ be the smallest such index, and let $I_{\max}$ be the greatest such index.
By construction, there exists a path from $(I_{\max},0)$ to $(I_{\min},0)$,
and every path from $\mathbf{v}_1$ to $\mathbf{v}_2$ must to through both $(I_{\max},0)$ and $(I_{\min},0)$.
Furthermore, if some vertex $(I,0)$ is reachable from $\mathbf{v}_1$, then
$I_{\max}$ is the greatest such integer.
Hence, one can check whether there exists a path from $\mathbf{v}_1$ to $(I_{\max},0)$,
and it remains to compute $I_{\min}$ and to check whether there exists a path from $(I_{\min},0)$ to $\mathbf{v}_2$.

By construction, there exist $|\Lambda(v_{I_{\min}})| \cdot 2^{|\Psi_{I_{\min}-1}|}$ nominal vertices $\mathbf{v}_3$ of the form
$(I_{\min},Z')$. For all such vertices $\mathbf{v}_3$, it is again easy to check whether $\mathbf{v}_3$ is reachable
from $(I_{\min},0)$ (since all paths from $(I_{\min},0)$ to $\mathbf{v}_3$ are short) and whether
$\mathbf{v}_2$ is reachable from $\mathbf{v}_3$ (since no path from $\mathbf{v}_3$ to $\mathbf{v}_2$ can use a vertex $(i,0)$).

It remains to treat the case where $\mathbf{v}_1$ and $\mathbf{v}_2$ are of no particular form.
There, either there exists a short path from $\mathbf{v}_1$ to $\mathbf{v}_2$,
or there exists nominal vertices $\mathbf{v}'_1$ and $\mathbf{v}'_2$ such that:
\begin{itemize}
 \item there exists a short path from $\mathbf{v}_1$ to $\mathbf{v}'_1$;
 \item there exists a short path from $\mathbf{v}'_2$ to $\mathbf{v}_2$;
 \item there exists a path from $\mathbf{v}'_1$ to $\mathbf{v}'_2$.
\end{itemize}
Enumerating all possible vertices $\mathbf{v}'_1$ and $\mathbf{v}'_2$
and checking whether such paths exist, using the method described above,
is feasible in \logspace.


\smallskip
We are left with computing the predicate $\pi_2$.
This latter computation is based on the
layered structure of the graph $\Gamma_G$ and of its labels.
Recall (from Section~\ref{subsec:reductiontodyck}) that
the Dyck paths are either sub-paths of nominal Dyck paths,
or generic Dyck paths, i.e. concatenations of nominal Dyck paths.

Hence, consider paths $\pi_1$ and $\pi_2$ in $\Gamma_G$, with
respective labels $\lambda_1$, $\lambda_2$,
and let $\pi$ be a sub-path of $\pi_1$.
If $\pi$ is a sub-path of a nominal (non-Dyck) path but
is not a sub-path of any nominal Dyck path, then,
by construction of the graph $\Gamma_G$,
the word $\lambda_1 \cdot \lambda_2$ cannot be a Dyck word.

From this remark, it follows that the word $\lambda_1 \cdot \lambda_2$ is a Dyck word if, and only if,
we are in one of the cases below:
\begin{itemize}
 \item both $\pi_1$ and $\pi_2$ are Dyck paths;
 \item there exists a sub-path $\rho$ of a nominal Dyck path such that both $\rho$ and $\pi_1 \cdot \rho \cdot \pi_2$ are Dyck paths
 (in particular, note that the source of $\rho$ must coincide with the sink of $\pi_1$,
 and that the sink of $\rho$ must coincide with the source of $\pi_2$);
 \item there exists paths $\rho_1$, $\rho_2$, $\rho_3$ and $\rho_4$ such that $\pi_1 = \rho_1 \cdot \rho_2$, $\pi_2 = \rho_3 \cdot \rho_4$
 and that all of $\rho_2$, $\rho_3$ and $\rho_1 \cdot \rho_4$ are Dyck paths.
\end{itemize}

Hence, the problem of computing $\pi_2$ is
\logspace-reducible to the problem of computing, for all vertices $\mathbf{v}_1,\ldots,\mathbf{v}_4$,
a Dyck path with source $\mathbf{v}_1$, sink $\mathbf{v}_2$, and that goes through $\mathbf{v}_3$ and $\mathbf{v}_4$
(if such a path exists).

Moreover, when $G = G_\epsilon$, for all nominal vertices $\mathbf{v}$,
there exists at most one nominal Dyck path with source $\mathbf{v}$;
that path is the unique minimal Dyck path with source $\mathbf{v}$ (if it exists).
This makes the problem of computing $\pi_2$ in \logspace~feasible,
which completes the
proof that the whole precomputation can be done in \logspace.


\medskip During the update phases, and whenever introducing or
deleting an edge $e$ in $G$, we must delete or introduce a bounded number
of edges in $\Gamma_G$; these edges are identified by \FO
formulas taking the edge $e$ into account.  Consequently, updating the
edge membership predicate of $\Gamma_G$ and the predicates $T$ and
$\pi_2$ can be done with \FO formulas.  Finally, deducing
whether $\sigma$ is winning from these predicates can be done using
\FO formulas again, which completes the proof of
Theorem~\ref{mainresult}.

\section{Dyck reachability with witness}
\label{app:dyck}\label{app-dyck}
Our dynamic algorithm for solving
$\Dyn\NonUnifParity_{\bfC,\kappa}$ goes through
an ad-hoc Dyck query over a graph. We present this problem, and then
analyze its dynamic complexity, which is $\Dyn\FO$.


Context-free graph query in acyclic graphs belong to the dynamic
complexity class
$\Dyn\FO$~\cite{MVZ15,WS07}.
We prove here a variant of this result.

Let $G=(V,E,L)$ with $E \subseteq V \times L \times V$ be a
  labelled acylic graph, whose labels are drawn from a set $L = L^+
  \cup L^= \cup L^-$, where $L^+$ and $L^-$ are in bijection with each
  other.  Let $x \mapsto \overline{x}$ be an involution of $L$ that
  maps $L^+$ to $L^-$ (and vice-versa) and fixes each element of $L^=$
  (those are ``neutral'' elements). 

  The Dyck reachability problem with witness asks for triples
  $(s,t,\rho)$, where $s$, $t$ are vertices of $G$ and $\rho$ is a
  path from $s$ to $t$ that is labelled by a string in the language
  $\mathbf{D}$ of \emph{Dyck words} built over the 
  grammar: $S \to \varepsilon \mid \ell \cdot S \cdot \overline{\ell}
  \cdot S \mid S \cdot \bullet \cdot S$, where $\ell \in L^+ \cup L^-$
  and $\bullet \in L^=$.

  Such a path may be viewed as a finite sequence $\mathbf{S} = s_0
  \cdot \ell_0 \cdot s_1 \cdot \ell_1 \cdot \ldots \cdot \ell_{k-1}
  \cdot s_k$, where $(s_i,\ell_i,s_{i+1}) \in E$ for all $i < k$, and
  we define its \emph{representative} as the predicate
  \[R^\mathbf{S}(u,\ell,v) \equiv u = v \vee \exists i \in
  \{0,\ldots,k-1\}, (u,\ell,v) = (s_i,a_i,s_{i+1}).\]

  The Dyck reachability problem with witness in acyclic graph is as
  follows: Given a labelled acyclic graph $G = (V,E,L)$ and two vertices
  $s$ and~$t$, is there a Dyck path with source $s$ and sink~$t$? If
  yes, give a predicate that represents such a path; if no, give an
  empty predicate. The dynamic version of that problem, denoted
  \Dyckwitness, assumes insertion and deletion of labelled edges (with
  the restriction that the graph should remain acyclic---which can be
  checked). We assume we start from an empty set of edges.

  \begin{thm}
    \label{thm:dyck-witness}
    Problem $\Dyckwitness$ can be solved in $\Dyn\FO$.
\end{thm}

Before developing a complete proof, we begin with an intuitive sketch.


\subsection{Sketch of the proof}

  Our proof is inspired by that of~\cite[Theorem
  7.2]{WS07}, which proves that checking whether Dyck
  paths exist, in the case $L^+ = \{a,b\}$, $L^= = \emptyset$ and $L^-
  = \{\overline{a},\overline{b}\}$, is in $\Dyn\FO$.

  Their proof is two-fold.  First, they introduce predicates
  $\pi_k(u_1,v_1,\ldots,u_k,v_k)$ that hold iff there exists paths
  from $u_i$ to $v_i$ with strings $s_i$ such that $s_1 \cdots s_k$ is
  a Dyck word, and compute each predicate $\pi_k$ as a first-order
  formula over the vocabulary formed of $\pi_2$, of the transitive
  closure of the graph, and of predicates for checking membership in
  $E$.

Second, the devise first-order formulas, over the same vocabulary augmented with predicates $\pi_k$,
that allow updating the value of $\pi_2$, from which they deduce a solution to the Dyck reachability problem.
Such formulas rely on the two following principles.
\begin{itemize}
 \item if $G'$ is the graph obtained from $G$ by adding an edge $(s,\lambda,t)$, then
 each (not necessarily Dyck) path $a \to b$ in $G'$ can be identified either with a path $a \to b$ in $G$
 or with a pair $(a \to s, t \to b)$ of paths in $G$;
 \item if $G'$ is the graph obtained from $G$ by deleting an edge $(s,\lambda,t)$, then
 each path $a \to b$ in $G'$ can be identified either with a path $a \to b$ in $G$ (if $s$ is not accesible from $a$ or
 $b$ is not accessible from $s$) or with a triple $(a \to x,(x,\mu,y),y \to b)$ where $a \to x$, $y \to b$ are paths in $G$
 and $(x,\mu,y)$ is an edge of $G'$ such that $s$ is accessible from $x$ but not from $y$.
\end{itemize}

Our adaptation is as follows.  First, we introduce auxiliary vertices
$\bullet$ and $\ell$ (for all labels $\ell \in L$) and labelled edges
$(\ell,\ell,\bullet)$ to $G$.  This shall allow us to select paths
with label $\ell$ by selecting the unique path that goes from $\ell$
to $\bullet$, in order to easily select and simulate the behavior of
the above-mentioned edges $(s,\lambda,t)$ and $(x,\mu,y)$, while
avoiding the dissymetry between labels $\ell$ and $\overline{\ell}$
that complicates the proofs of~\cite[Lemmas 7.4 and
7.5]{WS07}.

Second, while following the construction of the predicates $\pi_k$
of~\cite{WS07}, we also maintain in \FO predicates $q_k$
of arity $5k$ that, for each tuple $(u_1,v_1,\ldots,u_k,v_k)$
satisfying $\pi_k$, represent simultaneously paths $\rho_i$ from $u_i$
to $v_i$ with labels $s_i$ such that $s_1 \cdots s_k$ is a Dyck word.
The proof of~\cite[Lemma 7.3]{WS07} provides a recursive
procedure for checking whether collections $(\rho_1,\ldots,\rho_k)$
exist; we select unambiguously one such collection by maintaining a
linear order on those vertices that may be traversed by some path
$\rho_i$, following this recursive procedure and, whenever this
procedure uses a universal quantification on vertices satisfying a
property $\varphi$, we chose the minimal vertex satisfying $\varphi$.

Having constructed \emph{witnesses} for the predicates $\pi_k$, we
adapt the principles used in~\cite{WS07} for updating the
predicate $\pi_2$, thereby updating the predicate $q_2$ as well, from
which we deduce a solution to the Dyck reachability problem with
witness.

\subsection{Complete proof of Theorem~\ref{thm:dyck-witness}.}


We first assume that the sets $V$ and $L$ are disjoint
(up to replacing $V$ and $L$ by sets $\{(x,V) \mid x\in V\}$ and $\{(x,L) \mid x\in L\}$),
and consider some \emph{fresh} symbol $\bullet$, which does not belong to~$V$ nor to~$L$.

We replace now the graph $G = (V,E,L)$ by the new labelled acyclic graph $\mathcal{G} = (\mathcal{V},\mathcal{E},L)$
defined by: $\mathcal{V} = V \cup L \cup \{\bullet\}$ and $\mathcal{E} = E \cup \{(\ell,\ell,\bullet) \mid \ell \in L\}$.
Observe that $\mathcal{G}$ can be computed from $G$ and from a first-order formula,
and that inserting/deleting one labelled edge in $G$ amounts to inserting/deleting one labelled edge in $\mathcal{G}$.
Hence, we work on the graph $\mathcal{G}$.

\begin{figure}[!ht]
\begin{center}
\begin{tikzpicture}[scale=0.55]
\draw[draw=black!5,fill=black!5] (6.5,-4) -- (16.5,-4) -- (16.5,1.9) -- (6.5,1.9) -- cycle;

\SetGraphUnit{3}
\tikzset{EdgeStyle/.style = {->}}
\Vertex{1}
\EA(1){2}
\SO(1){3}
\EA(3){4}
\Edge[label={$a$,$b$},labelstyle={above}](1)(2)
\Edge[label=$a$,labelstyle={left}](1)(3)
\Edge[label=$\overline{a}$,labelstyle={left}](2)(4)
\Edge[label=$c$,labelstyle={above}](3)(4)
\Edge(1)(2)
\Edge(1)(3)
\Edge(2)(4)
\Edge(3)(4)

\tikzset{LabelStyle/.style= {fill=black!5}}

\EA[unit=6,L=$a$](2){a}
\SO[unit=1.5,L=$b$](a){b}
\SO[unit=1.5,L=$c$](b){c}
\EA[unit=5,L=$\overline{a}$](a){a2}
\SO[unit=1.5,L=$\overline{b}$](a2){b2}
\SO[unit=1.5,L=$\overline{c}$](b2){c2}
\EA[unit=2.5,L=$\bullet$](b){bul}

\Edge[label={$a$},style={bend left},labelstyle={above}](a)(bul)
\Edge[label={$b$}](b)(bul)
\Edge[label={$c$},style={bend right},labelstyle={below}](c)(bul)
\Edge[label={$\overline{a}$},style={bend right},labelstyle={above}](a2)(bul)
\Edge[label={$\overline{b}$}](b2)(bul)
\Edge[label={$\overline{c}$},style={bend left},labelstyle={below}](c2)(bul)
\Edge[style={bend left}](a)(bul)
\Edge[style={bend right}](c)(bul)
\Edge[style={bend right}](a2)(bul)
\Edge[style={bend left}](c2)(bul)

\node[anchor=north] at (1.5,-4) {$L = \{a,b,c,\overline{a},\overline{b},\overline{c}\}$};
\node[anchor=south] at (1.5,0.8) {Graph $G = (V,E,L)$};
\node[anchor=south] at (11.5,0.83) {Additional vertices and edges};
\end{tikzpicture}
\end{center}
\caption{From $G$ to $\mathcal{G}$ via adding vertices and edges}
\label{fig:G-to-calG(2)}
\end{figure}
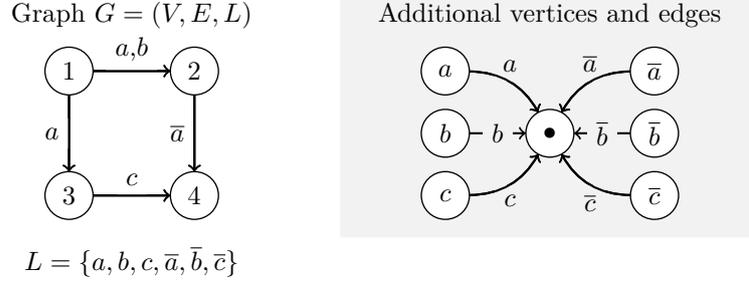

We recall now some facts that are proved in~\cite{WS07}.
First, observe that the edge relation $E$ can obviously be maintained
in \FO, since we track insertions and deletions of edges.

Second, the transitive closure of $\mathcal{G}$, ignoring the labels,
can be also maintained in \FO.  More precisely, in what follows, we
denote by $\mathcal{G}'$ the graph obtained from $\mathcal{G}$ by
inserting/deleting an edge and, for each predicate $X$ over
$\mathcal{G}$, we denote by $X'$ the analogous predicate over
$\mathcal{G}'$.  Here, denoting by $T$ the transitive closure in
$\mathcal{G}$, there exists update formulas:
\begin{itemize}
 \item $T'(u,v) \equiv T(u,v) \vee (T(u,s) \wedge T(t,v))$ after the insertion of some edge $(s,\ell,t)$;
 \item $T'(u,v) \equiv T(u,v) \wedge (T_1(u,v,s) \vee \exists (x,\mu,y) \in \mathcal{E}, T_2(u,v,x,\mu,y,s,\ell,t))$ after the deletion of some edge $(s,\ell,t)$,
\end{itemize}
where $T_1(u,v,s) \equiv T(v,s) \vee \neg T(u,s)$ and $T_2(u,v,x,\mu,y,s,\ell,t) \equiv 
(x,\mu,y) \neq (s,\ell,t) \wedge T(u,x) \wedge T(y,v) \wedge T(x,s) \wedge \neg T(y,s)$.
Indeed if $T_1(u,v,s)$ holds, then no path in $\mathcal{G}(u,v)$ can use the edge $(s,\ell,t)$;
if it does not hold, then $T'(u,v)$ holds if and only if there exits a path in $\mathcal{G}(u,v)$
that does not use the edge $(s,\ell,t)$ for exiting the set of vertices $\{z \in V \mid T(z,s)\}$,
to which $u$ belongs but $v$ does not belong.

\begin{figure}[!ht]
\begin{center}
\begin{tikzpicture}[scale=0.5]
\draw[draw=black!50,fill=black!15] (-4.5,4.5) -- (0,0) -- (4.5,4.5);

\SetGraphUnit{2}
\tikzset{EdgeStyle/.style = {->}}
\Vertex[L=$s$]{s}
\SOWE[L=$t$](s){t}
\Edge[label={$\lambda$},labelstyle={above left}](s)(t)

\NOEA[L=$x$](s){x}
\NOWE[L=$u$](x){u}
\SOEA[L=$y$](x){y}
\SOWE[L=$v$](y){v}

\Edge[label={$\mu$},labelstyle={above right}](x)(y)

\tikzset{EdgeStyle/.style = {densely dotted,->}}

\Edge[style={bend right}](u)(s)
\Edge[style={bend left}](u)(x)
\Edge[style={bend right}](x)(s)
\Edge(y)(v)
\end{tikzpicture}
\end{center}
\caption{Exiting the set $\{z \in V \mid T(z,s)\}$ (in gray) without using the edge $(s,\ell,t)$}
\label{fig:T-dans-dynFO(2)}
\end{figure}
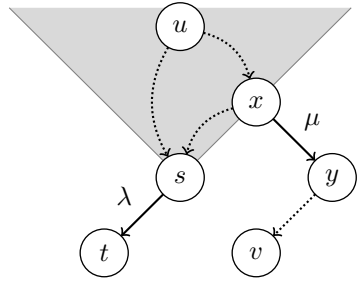

Third, consider the predicates $\pi_k$ defined as follows for all $k \geq 1$:
\[\pi_k(u_1,v_1,\ldots,u_k,v_k) \equiv \exists (\sigma_1,\ldots,\sigma_k) \in
\mathcal{G}(u_1,v_1) \times \ldots \times \mathcal{G}(u_k,v_k), [\sigma_1] \cdot \ldots \cdot [\sigma_k] \in \mathbf{D}.\]
It is proved in~\cite[Lemma~7.3]{WS07} that
all predicates $\pi_k$ belong to \FO[$\pi_2$] and, more precisely, that the following formulas hold for $k \geq 3$:
\[\pi_k(u_1,v_1,\ldots,u_k,v_k) \equiv
\bigvee_{i=2}^{k-1} \exists x,y,z \in \mathcal{V}, \Pi_k^i(u_1,v_1,\ldots,u_k,v_k,x,y,z),\]
where $\Pi_k^i$ is defined as:
\[\arraycolsep=1pt\begin{array}{lll}
\Pi_k(u_1,v_1,\ldots,u_k,v_k,x,y,z) \equiv \pi_2(x,v_1,u_2,z) & \wedge &
\pi_{k+2-i}(u_1,x,y,v_i,\ldots,u_k,v_k) \\ & \wedge & \pi_{i-1}(z,v_2,\ldots,u_i,y)
\end{array}\]

Third, we prove that $\pi_2$ can be maintained in \FO.  First, when
$G$ has no edge, $\pi_2$ can be computed from a \FO formula.  Then,
our proof is similar to that of~\cite{WS07}, with the
following difference: we will quantify over variables in $L$, whereas
such a quantification was implicit in~\cite{WS07}, where
$L$ was the fixed set $\{a,\overline{a},b,\overline{b}\}$.  Our
augmented structure $\mathcal{G}$ will mostly be useful here, since
$\ell \cdot \ell \cdot \bullet$ is the unique path in
$\mathcal{G}(\ell,\bullet)$ for all $\ell \in L$.

We first compute the updated version $\pi_2'$ of the predicate $\pi_2$ after inserting an edge $(s,\ell,t)$.
Consider a pair $(\mathbf{S}_1,\mathbf{S}_2)$ of paths in the updated graph $\mathcal{G}'$, such that $[\mathbf{S}_1] \cdot [\mathbf{S}_2] \in \mathbf{D}$.
Four cases are possible, depending on whether $\mathbf{S}_1$ and/or $\mathbf{S}_2$ use the edge $(s,\ell,t)$.
For each pair $(u,v) \in \mathcal{V}^2$,
every path $\mathbf{S} \in \mathcal{G}'(u,v)$ that uses the edge $(s,\ell,t)$ can be split into three smaller paths
$\mathbf{S}^1 \in \mathcal{G}(u,s)$, $\mathbf{S}^2 = s \cdot \ell \cdot t$ and $\mathbf{S}^3 \in \mathcal{G}(t,v)$,
such that $[\mathbf{S}] = [\mathbf{S}^1] \cdot [\mathbf{S}^2] \cdot [\mathbf{S}^3]$.
Finally, since $\mathbf{S}^2$ is not yet a path in $\mathcal{G}$, we replace it with the path $\ell \cdot \ell \cdot \bullet$,
whose label is also $\ell = [\mathbf{S}^2]$.
Hence, it follows that:
\[\arraycolsep=1pt\begin{array}{lll}\pi'_2(u_1,v_1,u_2,v_2) \equiv \pi_2(u_1,v_1,u_2,v_2) & \vee & \pi_4(u_1,s,\lambda,\bullet,t,v_1,u_2,v_2) \vee
\pi_4(u_1,v_1,u_2,s,\lambda,\bullet,t,v_2) \\
& \vee & \pi_4(u_1,s,\lambda,\bullet,t,v_1,u_2,s,\lambda,\bullet,t,v_2)\end{array}\]

Similarly, after deleting an edge $(s,\ell,t)$,
consider a pair $(\mathbf{S}_1,\mathbf{S}_2)$ of paths in the updated graph $\mathcal{G}$, such that $[\mathbf{S}_1] \cdot [\mathbf{S}_2] \in \mathbf{D}$.
For each pair $(u,v) \in \mathcal{V}^2$, a path $\mathbf{S} \in \mathcal{G}(u,v)$ still exists in $\mathcal{G}'$
if and only if either $T_1(u,v,s)$ holds, or if $\mathbf{S}$ does not use the edge $(s,\ell,t)$ to exit the set $\mathcal{Z} = \{z \in \mathcal{V} \mid T(z,s)\}$.
In the latter case, there exists some edge $(x,\mu,y)$ that exists $\mathcal{Z}$ and such that $\mathbf{S}$ can be split into three smaller paths
$\mathbf{S}^1 \in \mathcal{G}(u,x)$, $\mathbf{S}^2 = x \cdot \mu \cdot y$ and $\mathbf{S}^3 \in \mathcal{G}(y,v)$.
Here, observe that $\mathbf{S}^2$ may not be the only path in $\mathcal{G}(x,y)$, which might even contain paths using the edge $(s,\ell,t)$.
Hence, we ensure the choice of a path $\tilde{\mathbf{S}}^2$ with label $\mu = [\mathbf{S}^2]$ be choosing $\tilde{\mathbf{S}}^2$ inside the set
$\mathcal{G}(\mu,\bullet)$. It follows that
$\pi'_2(u_1,v_1,u_2,v_2) \equiv \psi_1 \vee \psi_2 \vee \psi_3 \vee \psi_4$, where:
\[\arraycolsep=1pt\begin{array}{lll}
\psi_1 & \equiv & \pi_2(u_1,v_1,u_2,v_2) \wedge T_1(u_1,v_1,s) \wedge T_1(u_2,v_2,s) \\
\psi_2 & \equiv & \exists (a,\mu,b) \in \mathcal{E}, T_2(u_1,v_1,a,\mu,b,s,\lambda,t) \wedge T_1(u_2,v_2,s) \wedge \pi_4(u_1,a,\mu,\bullet,b,v_1,u_2,v_2) \\
\psi_3 & \equiv & \exists (c,\nu,d) \in \mathcal{E}, T_1(u_1,v_1,x) \wedge T_2(u_2,v_2,c,\nu,d,x,\lambda,y) \wedge \pi_4(u_1,v_1,v_1,c,\nu,\bullet,d,v_2) \\
\psi_4 & \equiv & \exists (a,\mu,b),(c,\nu,d) \in \mathcal{E},
 T_2(u_1,v_1,a,\mu,b,x,\lambda,y) \wedge T_2(u_2,v_2,c,\nu,d,x,\lambda,y) \wedge \\
 & & \hspace*{34mm} \pi_6(u_1,a,\mu,\bullet,b,v_1,u_2,c,\nu,\bullet,d,v_2)
\end{array}\]

So far, we mostly rephrased the proof of~\cite{WS07}.  We
show now how witnesses can be extracted and maintained in \FO.

To that end, we first need to maintain a linear order $\leq$ (which we
may view as a binary predicate) on the a set $S$ that contains all
\emph{activated} vertices and labels of the \emph{original} graph $G$.
More precisely, we say that the elements of the set
\[\arraycolsep=1pt\begin{array}{lll}
\{\bullet\} & \cup & \{v \in V \mid \exists w \in V, \exists \ell \in L, (v,\ell,w) \in E \vee (w,\ell,v) \in E\} \\
& \cup & \{\ell,\overline{\ell} \in L \mid \exists v,w \in V, (v,\ell,w) \in E\}
\end{array}\]
are activated, and $S$ must contain all these elements.
However, as soon as this constraint is satisfied, the set $S$ and the order $\leq$ on $S$ are themselves arbitrary.

First, when $G$ contains no edge, we choose $S = \{\bullet\}$, and $\leq$ is the unique order on $S$.
In addition, for every vertex $z \in \mathcal{V}$ and every binary predicate $P$ over $\mathcal{V}$,
we denote by $\mathbf{Ext}[z,P]$ the binary predicate defined by:
$\mathbf{Ext}[z,P](u,v) \equiv P(u,v) \vee (v = z \wedge \neg P(z,z))$.
In the case where $P$ is a linear order on some subset $\Sigma$ of $\mathcal{V}$,
then $\mathbf{Ext}[z,P]$ either is equal to $P$, if $z \in \Sigma$, or is
obtained by extending $\leq$ to $\Sigma \cup \{z\}$, with $z$ as new greatest element.

Hence, after inserting an edge $(s,\ell,t)$, we choose $S' = S \cup \{s,t,\ell,\overline{\ell}\}$ and
we choose $\leq'$ to be the linear order on $S'$ defined as
$\leq' \equiv \mathbf{Ext}[s,\mathbf{Ext}[t,\mathbf{Ext}[\ell,\mathbf{Ext}[\overline{\ell},\leq]]]]$.
Conversely, after deleting an edge $(s,\ell,t)$, we do not change $S$ nor $\leq$.

Then, we follow the above proof that $\pi_2$ can be maintained in \FO,
but we do so in a constructive manner, maintaining witnesses for each
of the predicates $\pi_k$.  More precisely, for all $k \geq 1$, we
consider a predicate
$\rho_k(u_1,v_1,w_1,\lambda_1,x_1,\ldots,u_k,v_k,w_k,\lambda_k,x_k)$
that satisfies the following property $\mathcal{Q}_k$: ``\emph{for all
tuples $(u_1,v_1,\ldots,u_k,v_k)$ that satisfy $\pi_k$, there exists
paths $\mathbf{S}_1 \in \mathcal{G}(u_1,v_1),\ldots,\mathbf{S}_k \in
\mathcal{G}(u_k,v_k)$ such that $[\mathbf{S}_1] \cdot \ldots \cdot
[\mathbf{S}_k]$ is a Dyck word and such that
$\rho_k(u_1,v_1,w_1,\lambda_1,x_1,\ldots,u_k,v_k,w_k,\lambda_k,x_k)
\equiv \bigwedge_{i=1}^k R^{\mathbf{S}_i}(w_i,\lambda_i,x_i)$, where
we recall that $R^\mathbf{S}$ is the predicate that represents the
path~$\mathbf{S}$.}''

Note that, from such a predicate, recovering each of the predicates $R^{\mathbf{S}_i}$ is straightforward,
using universal quantification over the variables $w_j,\lambda_j,x_j$ for $j \neq i$.
In addition, it $\mathbf{S}_1$ is a path whose sink is the source of some other path $\mathbf{S}_2$,
then the concatenated path $\mathbf{S}_1 \cdot \mathbf{S}_2$ is represented by the predicate $R^{\mathbf{S}_1} \vee R^{\mathbf{S}_2}$.
Hence, it is harmless to say that the predicates $\rho_k$ provide us with \emph{tuples of paths}
instead of predicates representing such tuples, and
it will also be harmless to perform projections on such tuples or concatenation operations on pairs of words that share a sink and source.

Like the order $\leq$, such predicates $\rho_k$ are not uniquely defined. Hence,
we will use the order $\leq$ to build predicates $\rho_k$ unambiguously.
We prove now that such predicates $\rho_k$ belong to \FO[$\leq,\pi_2,\rho_2$].

We first choose
$\rho_1(u_1,v_1,w_1,\lambda_1,x_1) \equiv \rho_2(u_1,v_1,w_1,\lambda_1,x_1,u_1,u_1,u_1,\lambda_1,u_1)$.
We then treat the case $k \geq 3$ by induction,
noting that our task reduces to finding, using \FO formulas, a canonical tuple $(x,y,z,i)$ such that
$\Pi_k^i(u_1,v_1,\ldots,u_k,v_k,x,y,z)$ holds.
First, if $(x,y,z,i)$ is such a tuple, we know that each of $x$, $y$ and $z$ either are activated or belong to $\{u_1,v_1,\ldots,u_k,v_k\}$.
Indeed, if $x$ is not activated, then we must have $x = u_1 = v_1$, and similarly for $y$ and $z$.
Hence, we transform the order $\leq$ into the order
$\preccurlyeq~\equiv~\mathbf{Ext}[u_1,\mathbf{Ext}[v_1,\ldots,\mathbf{Ext}[u_k,\mathbf{Ext}[v_k,\leq]]]]$,
and we know that $x$, $y$, $z$ must belong to the domain of $\preccurlyeq$.

Consequently, there exists a minimal tuple $(x,y,z,i)$ (for the lexicographical order on $\mathrm{dom}(\preccurlyeq)^3 \times \{2,\ldots,k-1\}$)
such that $\Pi_k^i(u_1,v_1,\ldots,u_k,v_k,x,y,z)$ holds, and this tuple is definable in \FO[$\preccurlyeq,\pi_2,\ldots,\pi_{k-1}$], hence in \FO[$\leq,\pi_2$].
Then, by induction hypothesis, we know that $\rho_2$, $\rho_{k+2-i}$ and $\rho_{i-1}$
provide us with the path families
$(\tau_1,\tau_2) \in \mathcal{G}(x,v_1) \times \mathcal{G}(u_2,z)$,
$(\tau_0,\tau_5,\mathbf{S}_{i+1},\ldots,\mathbf{S}_k) \in \mathcal{G}(u_1,x) \times \mathcal{G}(y,v_i) \times
\mathcal{G}(u_{i+1},v_{i+1}) \times \ldots \times \mathcal{G}(u_k,v_k)$ and
$(\tau_3,\mathbf{S}_2,\ldots,\mathbf{S}_{i-1},\tau_4) \in \mathcal{G}(z,v_2) \times \mathcal{G}(u_3,v_3) \times \ldots \times \mathcal{G}(u_i,y)$.

Hence, we set $\mathbf{S}_1 = \tau_0 \cdot \tau_1$, $\mathbf{S}_2 = \tau_2 \cdot \tau_3$ and $\mathbf{S}_i = \tau_4 \cdot \tau_5$.
It follows that all of the predicates $R^{\mathbf{S}_1}, \ldots, R^{\mathbf{S}_k}$ belong to \FO[$\leq,\pi_2,\rho_1,\ldots,\rho_{k-1}$],
hence to \FO[$\leq,\pi_2,\rho_2$], and so does $\rho_k$.

Finally, we show that some predicate $\rho_2$ satisfying the property
$\mathcal{Q}_2$ can be maintained in \FO.  First, when $G$ does not
contain any edge, observe that each set $\mathcal{G}(u,v)$ is either
the singleton $\{u,u,v\}$, if $u \in L$ and $v = \bullet$, or the
empty set. Hence, there is a unique way of choosing $\rho_2$, which
can be done with first-order formulas.

Then, consider four vertices $u_1,v_1,u_2,v_2 \in \mathcal{V}$.
After inserting an edge $(s,\ell,t)$, deciding whether $\pi_2'(u_1,v_1,u_2,v_2)$ holds
can be done by considering four mutually exclusive cases:
\begin{enumerate}
 \item some paths $\mathbf{S}_1$ and $\mathbf{S}_2$,
 going from $u_1$ to $v_1$ and from $u_1$ to $v_2$, and such that $[\mathbf{S}_1] \cdot [\mathbf{S}_2] \in \mathbf{D}$, already exist in $\mathcal{G}$;
 in that case, $\rho_2$ provides us with two such paths $\mathbf{S}_1 \in \mathcal{G}(u_1,v_1)$ and $\mathbf{S}_2 \in \mathcal{G}(u_2,v_2)$;
 \item case \#1 is false, but such paths exist if we allow $\mathbf{S}_1$ to use the edge $(s,\ell,t)$;
 in that case, $\rho_4$ provides us with four paths
 $\tau_0 \in \mathcal{G}(u_1,s)$, $\tau_1 = \ell \cdot \ell \cdot \bullet$, $\tau_2 \in \mathcal{G}(t,v_1)$ and $\mathbf{S}_2 \in \mathcal{G}(u_2,v_2)$,
 and we set $\mathbf{S}_1 = \tau_0 \cdot (s \cdot \ell \cdot t) \cdot \tau_2$;
 \item cases \#1 and \#2 are false, but such paths exist if we allow $\mathbf{S}_2$ to use the edge $(s,\ell,t)$;
 in that case, $\rho_4$ provides us with four paths
 $\mathbf{S}_1 \in \mathcal{G}(u_1,v_1)$, $\tau_3 \in \mathcal{G}(u_2,s)$, $\tau_4 = \ell \cdot \ell \cdot \bullet$ and $\tau_5 \in \mathcal{G}(t,v_2)$,
 and we set $\mathbf{S}_2 = \tau_3 \cdot (s \cdot \ell \cdot t) \cdot \tau_5$;
 \item cases \#1 to \#3 are false, but such paths exist if we allow both $\mathbf{S}_1$ and $\mathbf{S}_2$ to use the edge $s,\ell,t)$;
 in that case, $\rho_6$ provides us with six paths $\tau_0,\ldots,\tau_5$, and we set $\mathbf{S}_1$ and $\mathbf{S}_2$ as above.
\end{enumerate}
Formulas in \FO[$\pi_2$] indicate in which of these cases we are,
then formulas in \FO[$\pi_2,\rho_2,\rho_4,\rho_6$] provide us with two paths $\mathbf{S}_1$ and $\mathbf{S}_2$
that are defined unambiguously, which proves that $\rho'_2$ may indeed be defined in \FO[$\leq,\pi_2,\rho_2$].

Similarly, after deleting an edge $(s,\ell,t)$, consider four vertices $u_1,v_1,u_2,v_2 \in \mathcal{V}$.
Deciding whether $\pi_2'(u_1,v_1,u_2,v_2)$ holds
can be done by considering four mutually exclusive cases, depending on whether $T_1(u_1,v_1,s)$ and/or $T_2(u_2,v_2,s)$ holds.
For instance, if $T_1(u_1,v_1,s)$ does not hold, there exists an edge $(a,\mu,b)$
such that $T_2(u_1,v_1,a,\mu,b,s,\ell,t)$ holds, and we can choose $(a,\mu,b)$ to be minimal
for the lexicographical order induced by $\leq$: this edge can be selected unambiguously
by using a formula in \FO[$T,\leq$].
Similarly, if $T_1(u_2,v_2,s)$ does not hold, we can choose unambiguously a minimal edge $(c,\nu,d)$
such that $T_2(u_2,v_2,c,\nu,d,s,\ell,t)$ holds.
This leads us to considering four mutually exclusive cases:
\begin{enumerate}
 \item if both $T_1(u_1,v_1,s)$ and $T_1(u_2,v_2,s)$ hold, then $\rho_2$ provides us with two paths
 $\mathbf{S}_1 \in \mathcal{G}(u_1,v_1)$ and $\mathbf{S}_2 \in \mathcal{G}(u_2,v_2)$ such that $[\mathbf{S}_1] \cdot [\mathbf{S}_2] \in \mathbf{D}$;
 \item if only $T_1(u_1,v_1,s)$ fails to hold, then $\rho_4$ provides us with four paths
 $\tau_0 \in \mathcal{G}(u_1,a)$, $\tau_1 = \mu \cdot \mu \cdot \bullet$, $\tau_2 \in \mathcal{G}(b,v_1)$ and $\mathbf{S}_2 \in \mathcal{G}(u_2,v_2)$,
 and we set $\mathbf{S}_1 = \tau_0 \cdot (a \cdot \mu \cdot b) \cdot \tau_2$;
 \item if only $T_1(u_2,v_2,s)$ fails to hold, then $\rho_4$ provides us with four paths
 $\mathbf{S}_1 \in \mathcal{G}(u_1,v_1)$, $\tau_3 \in \mathcal{G}(u_1,c)$, $\tau_4 = \nu \cdot \nu \cdot \bullet$ and $\tau_5 \in \mathcal{G}(b,v_2)$,
 and we set $\mathbf{S}_2 = \tau_3 \cdot (c \cdot \nu \cdot d) \cdot \tau_5$;
 \item if both $T_1(u_1,v_1,s)$ and $T_1(u_2,v_2,s)$ fail to hold, then $\rho_6$
 provides us with six paths $\tau_0,\ldots,\tau_5$, and we set $\mathbf{S}_1$ and $\mathbf{S}_2$ as above.
\end{enumerate}
Once again, it follows that $\rho'_2$ can be defined in \FO[$\leq,\pi_2,\rho_2$].

Overall, the predicates $T$, $\pi_2$, $\leq$ and $\rho_2$ can be
initialized and maintained with first-order formulas, i.e. can be
computed in \FO.  Hence, so can $\pi_1$ and $\rho_1$, which constitute
a solution of the Dyck reachability problem with witness.

\section{Computing \emph{uniform} winning strategies: Accessibility
  sets}
\label{section:decremental-III}\label{app-access}
In the core of the paper, for readability and page-limit concerns, we have
focused on 
a simple version of the dynamic parity problem, but the approach we
propose applies to the more demanding problem of synthesizing a
\emph{uniform} winning strategy, that is, the synthesis of a
$V_0$-selector which defines a winning strategy from \emph{every}
winning state of the game.

In the static (\textit{i.e.} non-dynamic) case, computing uniform
memoryless strategies in a parity game is known to be reducible to
computing several non-uniform winning strategies.  More precisely, let
$\sigma_1,\ldots,\sigma_k \in W_0$ be the winning states for the
player $P_0$, with associated non-uniform $V_0$-selectors
$f^1,\ldots,f^k$.  For each integer $i$, let $V^i$ be the set of those
states that are reachable from $\sigma_i$ when $P_0$ is bound to play
according to the $V_0$-selector $f^i$.

The following standard construction yields a uniform winning strategy
$g$ for $P_0$:
\begin{itemize}
\item for each vertex $s \in V$, let $i$ be the least integer, if any,
  such that $s \in V^i$: we set $g(s) = f^i(s)$, or we exclude $s$
  from the domain of $g$ if $s \notin \mathrm{dom}(f^i)$;
\item we exclude from the domain $\mathrm{dom}(g)$ each vertex $s \in
  V$ not yet treated.
\end{itemize}
Consequently, for every vertex $\sigma_i$ and every $V_0$-selector
$f^i$, being able to compute the set $V^i$ is crucial for computing
uniform winning strategies.  Such a computation will be made possible
by using \emph{accessibility sets}, which we define now.





Let $\sigma$ be the unique element of the set $\Psi_1$, let $g$ be a
$V_0$-selector, and let $i \in \{1,\ldots,\ell\}$ be some integer.  We
call \emph{accessibility set} of $g$, and denote by $\alpha_i^g$, the
subset of $\Psi_i = \bigcup_{v \in \Theta_i}\mathbf{T}(v)$ defined by:
$\alpha_i^g = \{s \in \Psi_i \mid \exists k \geq 0 \text{ and an
  extension } f \text{ of } g \text{ such that } f^k(\sigma) = s\}$,
where $f^k$ denotes the $k$-th iterate of $f$.

Unlike selector views, the accessibility sets of $g$ cannot be computed
directly in the same decremental fashion as above. However, once the selector
views $\omega_{i}^g$ and $\omega_i^{\ast g}$ are known, the~accessibility sets $\alpha_i^g$ can be
computed in an \emph{incremental} fashion.

Since $\Psi_1 = \{\sigma\}$, it is clear that $\alpha_1^g = \{\sigma\}$.
We compute then the sets $\alpha_i^g$ for $i \in \{2,\ldots,\ell\}$.
If $i$ is non-critical, then $\Psi_i \subseteq \Psi_{i-1}$, and therefore we have
$\alpha_i^g = \alpha_{i-1}^g \cap \Psi_i$.

Henceforth, we assume that $i \geq 2$ is critical, whence $\Psi_i = \Psi_{i-1} \uplus \{\theta_i\}$,
and computing $\alpha_i^g$ amounts to deciding whether $\theta_i \in \alpha_i^g$:
\begin{itemize}
 \item If $\theta_i \in \alpha_i^g$, let $f$ be an extension of $g$ such that $f^k(\sigma) = \theta_i$ for some integer $k \geq 1$,
 and let us assume that $k$ is minimal.
 
 If $\sigma \in \mathbf{T}(v_{i-1})$, then there exists a maximal integer $m \leq k-1$ such that $f^m(\sigma) \in \mathbf{T}(v_{i-1})$.
 Furthermore, if $\sigma \notin \mathbf{T}(v_{i-1})$, recall that $v_{i-1}$ must be the father of $v_i$, so that
 $\mathbf{T}(v_{i-1})$ is a separator of $V$ that leaves the vertices $\sigma$ and $\theta_i$
 in two distinct connected components of $V \setminus \mathbf{T}(v_{i-1})$.
 Hence, in that case too, there exists a maximal integer $m \leq k-1$ such that $f^m(\sigma) \in \mathbf{T}(v_{i-1})$.

 Let $s = f^m(\sigma)$. Then, we know that $s \in \mathbf{T}(v_{i-1}) \subseteq \mathbf{T}(v_i)$ and that $s \in \alpha_{i-1}^g$.
 By maximality of $m$ and minimality of $k$, the vertices $f^{m+1}(\sigma),\ldots,f^{k-1}(\sigma)$ all belong to $\Psi_{>i-1}(v_i)$.
 It follows that either $f(s) = \theta_i$ or $f_i(v_i,s) \in \{\theta_i\} \times \mathbb{N}$.
 
 \item Conversely, let us assume that there exists a vertex $s \in \mathbf{T}(v_i) \cap \alpha_{i-1}^g$
 such that $f(s) = \theta_i$ or $f_i(v_i,s) \in \{\theta_i\} \times \mathbb{N}$.
 In both cases, there exists a positive integer $\ell$ such that $f^\ell(s) = \theta_i$,
 and from now on we assume that $\ell$ is minimal.
 
 Let $m$ be an integer and $h$ be an extension of $g$ such that $h^m(\sigma) = s$.
 We define an extension $\varphi$ of $g$ by:
 \[\varphi\colon t \mapsto \begin{cases} f(t) & \text{if } t = f^u(s) \text{ for some integer } u \leq \ell-1 \\ h(t) & \text{otherwise}\end{cases}\]
 It is immediate that $\varphi$ is a well-defined extension of $g$ and that $\theta_i = \varphi^z(\sigma)$ for some integer $z \leq \ell+m$.
\end{itemize}

Moreover, for all vertices $s \in \mathbf{T}(v_i)$ and all extensions $f$ of $g$,
observe that $f(s) = \theta_i$ if and only if there exists sets $S, T \in 2^{\mathbf{T}(v_i)}$ such that
$s \in T$ and that $(f,S,T)$ is adapted to $g$ at step $i$ (the concept of adapted triples was
introduced at the end of the proof of Proposition~\ref{prop:decremental-II}).
In conclusion:
\begin{itemize}
 \item if there exists a function $\varphi \in \omega_i^{\ast g}$, sets $S, T \in 2^{\mathbf{T}(v_i)}$
 and a vertex $s \in T \cap \mathbf{T}(v_i) \cap \alpha_{i-1}^g$ such that $(\varphi,S,T)$ is relevant for $g$ at step $i$,
 then we have $\alpha_i^g = \alpha_{i-1}^g \cup \{\theta_i\}$;
 \item if there exists a function $\varphi \in \omega_i^{\ast g}$ and a vertex $s \in \mathbf{T}(v_i) \cap \alpha_{i-1}^g$
 such that $\varphi(s) \in \{\theta_i\} \times \mathbb{N}$, then we have $\alpha_i^g = \alpha_{i-1}^g \cup \{\theta_i\}$;
 \item otherwise, we have $\alpha_i^g = \alpha_{i-1}^g$,
\end{itemize}
where the sets of functions $\omega_i^{\ast g}(v)$ were introduced while proving Proposition~\ref{prop:decremental-II}
and were proved to be \FO-definable using $\omega_{i+1}^g$.
In particular, for \emph{all} integers $i \leq \ell-1$ (including $i = 1$ and non-critical indices $i$),
there exists a \FO-definable function $\mathbf{A}_i$ such that $\alpha_i^g = \mathbf{A}_i(\alpha_{i-1}^g,\omega_{i+1}^{g})$.

\section{Refining the graph of Subsection~\ref{subsec:reductiontodyck}}
\label{app:gammaG}\label{app-gammaG}
We target the resolution of the uniform version of our problem 
(see the discussion in 
Appendix~\ref{section:decremental-III}). We~thus formally construct an
extended graph with extra information for computing the accessibility
sets.


More precisely, we would like to compute the set
\[S =\bigcup_{1 \leq i \leq \ell} \alpha_i^g = \{s \in V \mid
\exists k \geq 0 \text{ and an extension } f \text{ of } g \text{ such that } f^k(\sigma) = s\}.\]
Since $\ell$ we proven to be non-critical, we have in fact
$S = \bigcup_{1 \leq i \leq \ell-1} \alpha_i^g$: we will rely below on this characterization.

We prove here that checking whether a vertex $\sigma
\in V$ is winning for $P_0$, selecting a winning $V_0$-selector $g$
for $P_0$ (if some exists) and then computing the set $S$
is reducible to the Dyck problem with witness in some graph $\Gamma_G$
which we define below, and which is a variant of the graph
constructed in Section~\ref{subsec:reductiontodyck}.
Recall Section~\ref{app:dyck}, where the Dyck
problem with witness is defined, and a dynamic solution to that
problem is provided.

We build an acyclic labelled graph $\Gamma_G$ as follows.
This graph is a variant of the graph described in Section~\ref{subsec:reductiontodyck},
whose structure has been modified to take into account accessibility sets and
to face issues that will be addressed in Appendix~\ref{app-extensions} below.

For all integers $i \leq \ell$, we denote by $\Lambda_i$
the set $\prod_{x \in \Theta_i} 2^{\Lambda(x)}$.
A nominal vertex of $\Gamma_G$ is a triple $(i,Z,A)$ with
$Z \in \Lambda_i$ and $A \subseteq \Psi_{i-1}$.

For all non-critical values of $i$, we add neutral edges:
$(i+1,Z,A) \xrightarrow{\bullet} (i,Z',A')$ where $Z = \Omega_i(Z)$ and $A' = A \cap \Psi_i$.

Then, let us order the vertices of the graph $(V,E)$ in some way, thereby endowing the set $V$
with a linear order $\leq$.
Consider some critical index $i$, and let $s_1 \leq \ldots \leq s_k$ be the vertices in $\mathbf{T}(v_i)$.
Let us also use a fresh symbol $s_{k+1}$.
Recall that a pair of vectors $(\mathbf{B}_{\theta_i,s})_{s \in \mathbf{T}(v_i)}$
and $ (\mathbf{B}'_{s,\theta_i})_{s \in \mathbf{T}(v_i)}$
is said to be suitable at step $i$ if their entries are non-empty subsets of $\{\bot,\top\}$.

Thus $\Gamma_G$ additionally contains the following edges and vertices,
for all sets $Z \in \Lambda_{i+1}$ and $A \subseteq \Psi_{i}$, all integers $a \in [i,k]$ and all vectors $\mathbf{B}$ and $\mathbf{B}'$ suitable at step $i$:
\begin{itemize}
\item $(i+1,Z,A) \xrightarrow{i,Z,A,\mathbf{B},\mathbf{B}'} (\theta_i,s_1,0)$;
\item $(\theta_i,s_a,0) \xrightarrow{\theta_i,s_a,\beta} (\theta_i,s_a,1) \xrightarrow{s_a,\theta_i,\mathbf{B},\mathbf{B}'} (\theta_i,s_{a+1},0)$
for all $\beta \in \choice(\theta_i,s_a)$;
\item $(s_a,\theta_i,0) \xrightarrow{s_a,\theta_i,\beta'} (s_a,\theta_i,1) \xrightarrow{s_a,\theta_i,\mathbf{B},\mathbf{B}'} (\theta_i,s_{a+1},0)$
for all $\beta' \in \choice(s_a,\theta_i)$;
\item $(\theta_i,s_{k+1},0) \xrightarrow{\bullet} (\theta_i,s_{k+1},\mathbf{B},\mathbf{B}',0)$;
\item $(\theta_i,s_{a+1},\mathbf{B},\mathbf{B}',0) \xrightarrow{\overline{s_a,\theta_i,\mathbf{B},\mathbf{B}'}} (s_a,\theta_i,\mathbf{B},\mathbf{B}',1)
\xrightarrow{\overline{s_a,\theta_i,\beta'}} (s_a,\theta_i,\mathbf{B},\mathbf{B}',0)$ for $\beta' = \mathbf{B}'_{s_a,\theta_i}$;
\item $(s_a,\theta_i,\mathbf{B},\mathbf{B}',0) \xrightarrow{\overline{\theta_i,s_a,\mathbf{B},\mathbf{B}'}} (\theta_i,s_a,\mathbf{B},\mathbf{B}',1)
\xrightarrow{\overline{\theta_i,s_a,\beta}} (\theta_i,s_a,\mathbf{B},\mathbf{B}',0)$ for $\beta = \mathbf{B}_{\theta_i,s_a}$;
\item $(\theta_i,s_1,\mathbf{B},\mathbf{B}',0) \xrightarrow{\overline{i,Z,A,\mathbf{B},\mathbf{B}'}} (i,Z',A')$ where $Z' = \Omega_i(Z,\mathbf{B},\mathbf{B}')$ and
$A = \mathbf{A}_i(A',Z)$.
\end{itemize}

In the sequel, we call \emph{nominal path} every labelled path of the form
$(i+1,Z) \rightarrow (i,Z')$ or
    \begin{multline*}
      (i+1,Z,A) \rightarrow (\theta_i,s_1,0) \rightarrow (\theta_i,s_1,1) \rightarrow (s_1,\theta_i,0) \rightarrow (s_1,\theta_i,1) \rightarrow \ldots \\
      \ldots \rightarrow (s_k,\theta_i,1) \rightarrow (\theta_i,s_{k+1},0) \rightarrow (\theta_i,s_{k+1},\mathbf{B},\mathbf{B}',0)
      \rightarrow (s_k,\theta_i,\mathbf{B},\mathbf{B}',1) \rightarrow \ldots \\
      \ldots \rightarrow (s_1,\theta_i,\mathbf{B},\mathbf{B}',0) \rightarrow (\theta_i,s_1,\mathbf{B},\mathbf{B}',1)
      \rightarrow (\theta_i,s_1,\mathbf{B},\mathbf{B}',0) \rightarrow (i,Z',A').
    \end{multline*}
Due to the layered structure of the graph $\Gamma_G$, every Dyck path in $\Gamma_G$ is either a sub-path of a nominal Dyck path or
of concatenation of nominal Dyck paths; in the latter case, we call the path a \emph{generic Dyck path}.

Hence, let $\mathcal{P}_\Gamma$ be the class of generic Dyck paths in $\Gamma_G$ with source in
$\{(\ell,\Omega_\ell,A) \mid A \subseteq \Psi_{\ell-1}\}$
and with sink in $\{(1,Z,\emptyset) \mid Z \subseteq \Lambda(v_1)\}$. For each path $\pi \in \mathcal{P}_\Gamma$, we~let
    \begin{itemize}
    \item $Z(\pi)$ be the subset of $\Lambda(v_1)$ such that $(1,Z(\pi),\emptyset)$ is the sink of~$\pi$;
    \item $E(\pi)$ be the subset of $E$ defined as
 \[\arraycolsep=1pt\begin{array}{lll}E(\pi) & = &
   \left\{(\theta_i,s) \left|\begin{array}{l} i \text{ is critical, } \theta_i \in V_0 \text{, } s \in \mathbf{T}(v_i) \text{ and } \\ \pi \text{ goes through a vertex }
   (\theta_i,s_{k+1},\mathbf{B},\mathbf{B}',0) \text{ with } \mathbf{B}_{\theta_i,s} = \{\top\}\end{array}\right\}\right. \cup \\
   & & \left\{(s,\theta_i) \left|\begin{array}{l} i \text{ is critical, } s \in V_0 \cap \mathbf{T}(v_i) \text{ and } \\ \pi \text{ goes through a vertex }
   (\theta_i,s_{k+1},\mathbf{B},\mathbf{B}',0) \text{ with } \mathbf{B}'_{s,\theta_i} = \{\top\}\end{array}\right\}\right.;\end{array}\]
    \item $V(\pi)$ be the subset of $V$ defined as $V(\pi) = \bigcup\{A \subseteq V \mid \pi \text{ uses some vertex } (i,Z,A)\}$.
\end{itemize}

The following result follows directly from the above classification of
Dyck paths in $\Gamma_G$ and from Sections~\ref{section:decremental-II}
and~\ref{section:decremental-III} (it~refines Proposition~\ref{prop:accessibility}).

\begin{proposition}
  \label{prop:accessibilityII}
  For every $V_0$-selector $g$, there exists a unique Dyck path $\pi^g
  \in \mathcal{P}_\Gamma$ such that $E(\pi^g) = \{(s,g(s)) \mid s \in
  \mathrm{dom}(g)\}$.  The function $g \mapsto \pi^g$ is a
  bijection from the set of $V_0$-selectors to $\mathcal{P}_\Gamma$.
  Moreover, we have $Z(\pi^g) = \omega_i^g$, i.e. $g$ ensures
  $P_0$'s victory when the play is in~$\sigma$ if, and only~if,
  $Z(\pi^g) = \{\sigma \mapsto \top\}$.  In~addition, the set
  $V(\pi^g)$ is equal to the union of accessibility sets $\bigcup_{1
    \leq i \leq \ell-1}\alpha_i^g$.
\end{proposition}

Consequently, answering the non-uniform two-player parity game problem
amounts to checking if there exists a Dyck path $\pi \in \mathcal{P}_\Gamma$
such that $Z(\pi) = \{\sigma \mapsto \top\}$ and, if yes, to
computing the sets $E(\pi)$ and $V(\pi)$.

\section{Refining the complexity analysis}
\label{app:complexity-2}\label{app-overall-2}
While the results of Appendix~\ref{app:complexity} hold
in the framework of Theorem~\ref{thm-main},
we need to extend them to a broader framework,
where (a) witnesses are maintained, which forces us to precompute
and maintain an additional predicate $\rho_2$ (see Appendix~\ref{app-dyck}), and
(b) the underlying auxiliary graph is that of Appendix~\ref{app:gammaG}.

This proof extension is immediate, due to the following remarks:
\begin{itemize}
 \item the precomputation of the vertices and edges of $\Gamma_G$ is exactly the same;
 \item we only know that paths of length $32(\kappa+1)+3$ (and not only $3$ like in Appendix~\ref{app:complexity})
 contain nominal vertices, hence we modify the notion of \emph{short} path
 to include all paths of length at most $32(\kappa+1)+3$;
 \item all the other arguments used to prove that $T$ or $\pi_2$ can be precomputed need no further change,
 and those used for computing $\pi_2$ also allows us to compute $\rho_2$;
 \item it is still straightforward that edge membership in $\Gamma_G$
 and the predicates $T$, $\pi_2$ and $\rho_2$ can be updated with \FO formulas.
\end{itemize}

\section{From the non-uniform to the uniform parity game problem, with
  additional update operations}
\label{app-extensions}
We finish by 
extending the scope of Theorem~\ref{thm-main} in two different directions:
first we extend our algorithm to solve parity games \emph{uniformly} (using
the construction of Appendix~\ref{app-gammaG}),
computing all the winning states and a uniform winning strategy from this
winning region. We~then add new operations by which the color and owner of a
state of the parity game can be modified.

\subsection{Towards uniformity}

We first consider \emph{uniformity}, showing that 
computing both the set $W_0$ of vertices that are winning for $P_0$
and a uniform $V_0$-strategy $f_0$ that leads to $P_0$'s victory
whenever the play is in~$W_0$ can be performed in $\Dyn[\logspace]\FO$,
provided that the only update operations are insertions and deletions
of edges.

Indeed, thanks to
Theorem~\ref{thm-main}, for each vertex $\sigma \in
V$, we can check whether $\sigma \in W_0$ and, if yes, we can compute
a $V_0$-strategy $g^\sigma$ and the set $V^\sigma$ of those vertices
that are accessible from $\sigma$ if $P_0$ is bound to play according
to the strategy $g^\sigma$. For every $\sigma$, we write
$\Gamma^\sigma$ for the acyclic graph computed for $\sigma$ in
Appendix~\ref{app:gammaG} (sketched version in
Subsection~\ref{subsec:reductiontodyck}).
Likewise, we denote by $\mathcal{D}^\sigma = (\mathcal{T}^\sigma,\mathbf{T}^\sigma)$
the tree decomposition rooted at $\sigma$, by
$v_i^\sigma$ the $i$-th element of the 
depth-first traversal of $\mathcal{T}^\sigma$, and so on.

Performing these computations for all vertices $\sigma$, with \FO
updates, can be done as follows.  For each vertex $\sigma \in V$,
Proposition~\ref{prop:accessibility} states that computing $g^\sigma$
and $V^\sigma$ is reducible to the Dyck accessibility problem with
witnesses on some acyclic graph $\Gamma^\sigma$.  Computing naively
$g^\sigma$ and $V^\sigma$ in parallel would result, for each update,
in adding one edge and deleting another edge in each of the $|V|$
graphs $\Gamma^\sigma$, and this would result in a polynomial number
of \FO computations, which might not be in \FO.

Instead we proceed as follows.  
First, we unite the graphs $\Gamma^\sigma$ in one single graph
$\Gamma$, replacing each vertex $\mathbf{v}$ of
$\Gamma^\sigma$ by a pair $(\sigma,\mathbf{v})$,
and replacing each each label $\lambda$ or $\overline{\lambda}$
(with the exlusion of neutral labels $\bullet$)
by a pair $(\sigma,\lambda)$ or $\overline{\sigma,\lambda}$.

Doing so, we have juxtaposed $|V|$ disjoint graphs $\Gamma^\sigma$
into one single graph $\Gamma$.

Then, we glue some vertices of these graphs $\Gamma^\sigma$.
For each pair $(s,t) \in V^2$ of vertices of our initial graph $G = (V,E)$,
some graphs $\Gamma^\sigma$ contain nodes $(\sigma,s,t,0)$ and $(\sigma,s,t,1)$:
we merge the nodes $(\sigma,s,t,0)$ into one single node $(s,t,0)$,
and the nodes $(\sigma,s,t,1)$ into one single node $(s,t,1)$.
Similarly, for each label for all $\beta \in \choice(s,t)$,
we merge the labels $(\sigma,s,t,\beta)$ into one single label
$(s,t,\beta)$, and the labels $\overline{\sigma,s,t,\beta}$ into
one single label $\overline{s,t,\beta}$.

The new labelled edges $(s,t,0) \rightarrow (s,t,1)$ play the role of hubs.
They ensure that, if one edge $(s,t)$ of~$G$ is changed, then only
one labelled edge $(s,t,0) \xrightarrow{s,t,\beta} (s,t,1)$ is removed from $\Gamma$,
and one labelled edge $(s,t,0) \xrightarrow{s,t,\beta'} (s,t,1)$ is added instead.


  These merges may introduce cycles inside the graph $\Gamma$, but we may
  still view each graph~$\Gamma^\sigma$ as an acyclic subgraph of~$\Gamma$,
  whose vertices either are merged vertices or are of the form
  $(\sigma,\mathbf{v})$, and whose edges are either hubs, or have a label
  $\bullet$ or $(\sigma,\lambda)$ or $\overline{\sigma,\lambda}$.
  Hence, we may compute $g^\sigma$ and $V^\sigma$ in
  parallel by using the graph~$\Gamma^\sigma$, although each edge insertion or
  deletion in~$G$ results in only one edge insertion and one edge deletion
  in~$\Gamma$.

  Finally, notice that a naive implementation of the approach above
  would require updating polynomially-many predicates at each
  step. However, those $|V|$ predicates (of~arity~$k$,~say) would be
  updated uniformly, and could be represented as a single predicate of
  arity~$k+1$ (with one extra variable for~$\sigma$) that would be
  updated by a single \FO formula.

\smallskip
  This already provides us with a way to compute the set $W_0$ itself.
  Moreover, we derive from the set~$W_0$ and the strategies~$g^\sigma$
  and sets~$V^\sigma$ the following uniform winning strategy~$g$,
  which we already described in Section~\ref{section:decremental-III}. 
  Letting~$\leq$ be an arbitrary linear order on~$V$:
\begin{itemize}
 \item for each vertex $s \in V$, let $\sigma$ be the least vertex of $W_0$ (for the order $\leq$) such that
 $s \in V^\sigma$, if such a vertex exists, and we set $g(s) = g^\sigma(s)$, or we exclude $s$ from the domain of $g$ if $s \notin \mathrm{dom}(g^\sigma)$;
 \item we exclude from the domain $\mathrm{dom}(g)$ each vertex $s \in V$ not yet treated.
\end{itemize}
Such a computation can be performed thanks to a \FO formula, which
indeed proves that the set $W_0$ and the function $g$ can be computed
by using \FO updates. 


\subsection{Changing Vertex Ownership and Color}

We show now that allowing to transfer a vertex from $V_0$ to $V_1$ or vice-versa as an update operation
does not increase the dynamic complexity of the uniform two-player parity game problem.
Indeed, instead of studying the graph $G = (V,E)$ endowed with a partition of $V$ into two sets $V_0$ and $V_1$
and with a coloring function $c\colon V \to \mathbb{N}$,
we first study the auxiliary graph $\mathbf{G} = (\mathbf{V},\mathbf{E})$ with a partition of $\mathbf{V}$ into sets $\mathbf{V}_0$ and $\mathbf{V}_1$
and a coloring function $\mathbf{c}\colon \mathbf{V} \to \mathbb{N}$ defined as follows:
\begin{itemize}
 \item $\mathbf{V} = V \times \{0,1,2\}$;
 \item $\mathbf{V}_0 = V \times \{0\}$, and $\mathbf{V}_1 = V \times \{1,2\}$;
 \item $\mathbf{c}\colon (s,i) \mapsto c(s)$ for all $s \in V$ and $i \in \{0,1,2\}$;
 \item $\mathbf{E} = \{((s,i),(t,2)) \mid i \in \{0,1\} \text{ and } (s,t) \in E\} \cup \{((s,2),(s,i)) \mid s \in V_i\}$.
\end{itemize}

If $\mathcal{D} = (\mathcal{T},\mathbf{T})$ is a tree decomposition of $G$ of width $\kappa$, then
$\mathcal{D}' = (\mathcal{T},\mathbf{T}')$ is a tree decomposition of $\mathbf{G}$ of width at most $3\kappa+2$, where
$\mathbf{T}'$ is defined by $\mathbf{T}'\colon v \mapsto \mathbf{T}(v) \times \{0,1,2\}$.
Moreover, transferring a vertex $s$ from $V_i$ to $V_{1-i}$ in $G$
amounts to deleting the edge $((s,2),(s,i))$ and adding the edge $((s,2),(s,1-i))$ in $\mathbf{G}$.

This means that dynamically solving the uniform two-player parity game
problem with edge insertions\slash deletions and 
transfer of vertices between $V_0$ and~$V_1$ is bfo-reducible to solving
the same problem with only edge insertions\slash deletions.

Similarly, further allowing modifications the color $c(s)$ of a vertex $s \in V$ as an update operation
does not increase the dynamic complexity of the uniform two-player parity game problem.
This time, instead of studying the graph $G = (V,E)$ endowed with a partition of $V$ into two sets $V_0$ and $V_1$
and with a coloring function $c\colon V \to \{1,\ldots,\mathbf{C}\}$,
we first study the auxiliary graph $\mathbf{G} = (\mathbf{V},\mathbf{E})$ with a partition of $\mathbf{V}$ into sets $\mathbf{V}_0$ and $\mathbf{V}_1$
and a coloring function $\mathbf{c}\colon \mathbf{V} \to \mathbb{N}$ defined as follows:
\begin{itemize}
 \item $\mathbf{V} = V \times \{-1,0,1,\ldots,\mathbf{C}\}$;
 \item $\mathbf{V}_0 = V_0 \times \{-1,0,1,\ldots,\mathbf{C}\}$, and $\mathbf{V}_1 = V_1 \times \{-1,0,1,\ldots,\mathbf{C}\}$;
 \item $\mathbf{c} \colon (s,i) \mapsto \max\{i,1\}$ for all $s \in V$ and $i \in \{-1,0,1,\ldots,\mathbf{C}\}$;
 \item $\mathbf{E} = \{((s,0),(t,-1)) \mid (s,t) \in E\} \cup \{((s,-1),(s,i)) \mid s \in V, 1 \leq i \leq \mathbf{C}\} \cup 
 \{((s,c(s)),(s,0)) \mid s \in V\}$.
\end{itemize}

If $\mathcal{D} = (\mathcal{T},\mathbf{T})$ is a tree decomposition of $G$ of width $\kappa$, then
$\mathcal{D}' = (\mathcal{T},\mathbf{T}')$ is a tree decomposition of $\mathbf{G}$ of width at most $(\mathbf{C}+2)(\kappa+1)-1$, where
$\mathbf{T}'$ is defined by $\mathbf{T}' \colon v \mapsto \mathbf{T}(v) \times \{-1,0,1,\ldots,\mathbf{C}\}$.
Moreover, replacing the value~$u$ of the color of a vertex~$s$ by the new value~$v$
amounts to deleting the edge $((s,u),(s,0))$ and adding the edge $((s,v),(s,0))$ in~$\mathbf{G}$.

%
%
%

\makeatletter
\let\@savethebibliography\thebibliography
\def\thebibliography#1{%
  \@savethebibliography{#1}
  \global\def\c@enumiv{\value{suitebib}}}
\makeatother

\putbib[nmbib]
\end{bibunit}
\fi

\end{document}